\documentclass[a4paper,10pt]{article}

\usepackage{fullpage}
\usepackage{xspace}
\usepackage{comment}

\usepackage{amsmath,amssymb,amstext,amsthm}

\newtheorem{theorem}{Theorem}
\newtheorem{lemma}[theorem]{Lemma}
\newtheorem{definition}[theorem]{Definition}

\newcommand{\ctw}{\mathbf{ctw}}
\newcommand{\pw}{\mathbf{pw}}

\newcommand{\FF}{\ensuremath{\mathfrak{F}}}

\newcommand{\Gg}{\ensuremath{\mathcal{G}}}
\newcommand{\UU}{\ensuremath{\mathfrak{U}}}
\newcommand{\ptrace}{trace}
\newcommand{\ptraces}{traces}

\newcommand{\vf}[2]{\mathcal{V}_{(#1,#2)}}

\newcommand{\cutfam}{\mathcal{N}}
\newcommand{\bipclass}{\mathcal{B}}
\newcommand{\match}{\mathfrak{M}}

\newcommand{\buss}{\mathfrak{B}}
\newcommand{\rev}{{\rm{rev}}}

\title{Computing cutwidth and pathwidth of semi-complete digraphs via degree orderings}

\date{}

\author{Micha\l{} Pilipczuk\thanks{Department of Informatics, University of Bergen, Norway, \texttt{michal.pilipczuk@ii.uib.no}}}

\begin{document}

\maketitle

\begin{abstract}
The notions of cutwidth and pathwidth of digraphs play a central role in the containment theory for tournaments, or more generally semi-complete digraphs, developed in a recent series of papers by Chudnovsky, Fradkin, Kim, Scott, and Seymour~\cite{ChudnovskyFS2011,ChudnovskySS2011,ChudnovskyS11,fradkin-seymourEDP,fradkin-seymour,kim-seymour-minors}. In this work we introduce a new approach to computing these width measures on semi-complete digraphs, via degree orderings. Using the new technique we are able to reprove the main results of~\cite{ChudnovskyFS2011,fradkin-seymour} in a unified and significantly simplified way, as well as obtain new results. First, we present polynomial-time approximation algorithms for both cutwidth and pathwidth, faster and simpler than the previously known ones; the most significant improvement is in case of pathwidth, where instead of previously known $O(OPT)$-approximation in fixed-parameter tractable time~\cite{my} we obtain a constant-factor approximation in polynomial time. Secondly, by exploiting the new set of obstacles for cutwidth and pathwidth, we show that topological containment and immersion in semi-complete digraphs can be tested in single-exponential fixed-parameter tractable time. Finally, we present how the new approach can be used to obtain exact fixed-parameter tractable algorithms for cutwidth and pathwidth, with single-exponential running time dependency on the optimal width.
\end{abstract}

\section{Introduction}

{\bf{Minors of (di)graphs.}} The {\em{Graph Minors}} series of Robertson and Seymour not only resolved the Wagner's Conjecture, but also provided a number of algorithmic tools for investigating topological structure of graphs. A useful concept that is repeatedly used in many arguments in this theory, is the WIN/WIN approach that can be formulated as follows. We introduce a graph width measure $\mu$, usually defined as the optimal width of some decomposition. The crucial part is to prove a structural theorem that asserts that graphs with large measure $\mu$ contain certain obstacles to admitting a simpler decomposition. These obstacles can be exploited in many ways: either we can immediately provide an answer to the algorithmic problem under consideration, or, for instance, find a vertex in the obstacle whose deletion does not change the answer. If no obstacle is present, we can find a decomposition of small width and try to apply algorithms based on the dynamic programming paradigm. The most classical measure $\mu$ is the treewidth of the graph, accompanied by a grid minor as the obstacle; however, the framework has found applications in many other contexts as well.

Since the beginning of the graph minors project there was a thrilling question, to what extent similar results can be obtained in the directed setting. The answer is still unclear; however, it seems that there is no hope for a positive outcome in the general case. There are several reasons for this. Despite numerous attempts, we still lack any good width measure that would be an analogue of treewidth~\cite{ganian-i-kumple}; moreover, most of the containment problems are already NP-hard for small, constant-size graphs $H$ that are to be embedded into input graph $G$~\cite{FortuneHW80}. Therefore, the scope of research moved to identifying subclasses of digraphs where construction of a sound containment theory could be possible. We invite an interested reader to the introduction of a recent work of Fomin and the current author~\cite{my} for a more detailed overview of the status of containment problems in digraphs.

\vskip 0.15cm

\noindent{\bf{Tournaments.}} In a recent series of papers, Chudnovsky, Fradkin, Kim, Scott, and Seymour~\cite{ChudnovskyFS2011,ChudnovskySS2011,ChudnovskyS11,fradkin-seymourEDP,fradkin-seymour,kim-seymour-minors} identify tournaments as a class where an elegant containment theory can be constructed. Recall that a tournament is a digraph where every two vertices are connected by an arc directed in one of two possible directions; the results hold also for a slightly more general class of {\em{semi-complete}} digraphs, where we additionally allow existence of two arcs directed in opposite directions. The central notions of the theory are two width measures of digraphs: {\em{cutwidth}} and {\em{pathwidth}}. The first one is based on the ordering of vertices and resembles classical cutwidth in the undirected setting~\cite{thilikos-i-kumple}, with the exception that only arcs directed forward in the ordering contribute to the width of a cut. The second one is a similar generalization of undirected pathwidth. See Section \ref{sec:width} for formal definitions.

Chudnovsky, Fradkin, and Seymour~\cite{ChudnovskyFS2011} prove a structural theorem that provides a set of obstacles for admitting an ordering of small (cut)width; a similar theorem for pathwidth was proven by Fradkin and Seymour~\cite{fradkin-seymour}. A large enough obstacle for cutwidth admits every fixed-size digraph as an immersion, and the corresponding is true also for pathwidth and topological containment. Basing on the first result, Chudnovsky and Seymour~\cite{ChudnovskyS11} were able to show that immersion is a well-quasi-order on the class of semi-complete digraphs. The same is not true for topological containment, but recently Kim and Seymour~\cite{kim-seymour-minors} introduced a relaxed notion of {\em{minor}} order, which indeed is a well-quasi-order of semi-complete digraphs.

As far as the algorithmic aspects of the work of Chudnovsky, Fradkin, Kim, Scott, and Seymour are concerned, the original proofs of the structural theorems can be turned into approximation algorithms, which given a semi-complete digraph $T$ and an integer $k$ find a decomposition of $T$ of width $O(k^2)$, or provide an obstacle for admitting decomposition of width at most $k$. For cutwidth the running time is polynomial, but for pathwidth it is $O(|V(T)|^{m(k)})$ for some function $m$; this excludes usage of this approximation as a subroutine in any FPT algorithm\footnote{An algorithm for a parameterized problem is said to be {\em{fixed-parameter tractable}} (FPT), if it runs in time $f(k)\cdot n^c$ on instances of size $n$ and parameter $k$, for some function $f$ and constant $c$.}, for instance, for topological containment testing. This gap has been bridged by Fomin and the current author~\cite{my} by presenting an FPT approximation algorithm for pathwidth that again either finds a path decomposition of $T$ of width $O(k^2)$ or provides an obstacle for admitting path decomposition of width at most $k$, but runs in $O(2^{O(k\log k)}|V(T)|^3\log |V(T)|)$ time. The approach in~\cite{my} is based on replacing the crucial part of the algorithm of Fradkin and Seymour that was implemented by a brute-force enumeration, by a more refined argument based on the colour coding technique.

As far as computation of optimal decompositions is concerned, by the well-quasi-ordering result of Chudnovsky and Seymour~\cite{ChudnovskyS11} we have that the class of semi-complete digraphs of cutwidth bounded by a constant is characterized by a finite set of forbidden immersions; the result of Kim and Seymour~\cite{kim-seymour-minors} proves that we can infer the same conclusion about pathwidth and minors. Having approximated the corresponding parameter, in FPT time we can check if any of these forbidden structures is contained in a given semi-complete digraph, using dynamic programming. This gives FPT exact algorithms computing cutwidth and pathwidth; however, they are both non-uniform ---  the algorithms depend on the set of forbidden structures which is unknown --- and non-constructive --- they provide just the value of the width measure, and not the optimal decomposition. To the best of author's knowledge, nothing better was known for these problems.

\vskip 0.15cm

\noindent{\bf{Our results and techniques.}} In this paper we present a new approach to computing cutwidth and pathwidth of semi-complete digraphs, via degree orderings. Using the new technique we are able to reprove the structural theorems for both parameters in a unified and simplified way, obtaining in both cases polynomial-time algorithms and, in case of pathwidth, a constant-factor approximation. The technique can be also used to develop FPT exact algorithm for both width measures with single-exponential dependency of the running time on the optimal width, as well as to trim the dependency of the running time on the size of the tested digraph in the topological containment and immersion tests to single-exponential. All the algorithms presented in this paper have quadratic dependency on the number of vertices of a given semi-complete digraph; note that this is {\em{linear}} in the input size. 

\begin{figure}
\begin{centering}
\noindent\begin{tabular}{| c | c | c |}
\hline
Problem & Previous results & This work \\
\hline
Cutwidth approximation & $O(n^3)$ time, width $O(k^2)$~\cite{ChudnovskyFS2011} & $O(n^2)$ time, width $O(k^2)$ \\[0.2cm]
Cutwidth exact & $O(f(k)\cdot n^3)$ time, & $O(2^{O(k)}\cdot n^2)$ time \\
& non-uniform, non-constructive~\cite{ChudnovskyFS2011,ChudnovskyS11} & \\[0.2cm]
Pathwidth approximation & $O(2^{O(k\log k)} \cdot n^3\log n)$ time,  & $O(kn^2)$ time,  \\
& width $O(k^2)$~\cite{my} & width $7k$ \\[0.2cm]
Pathwidth exact & $O(f(k)\cdot n^3\log n)$ time, & $O(2^{O(k\log k)}\cdot n^2)$ time \\
& non-uniform, non-constructive~\cite{my,kim-seymour-minors} & \\[0.2cm]
Immersion & $O(f(|H|)\cdot n^3)$ time~\cite{ChudnovskyFS2011} & $O(2^{O(|H|^2\log |H|)}\cdot n^2)$ time \\[0.2cm]
Topological containment & $O(f(|H|)\cdot n^3\log n)$ time~\cite{my} & $O(2^{O(|H|\log |H|)}\cdot n^2)$ time \\
\hline
\end{tabular}
\caption{Comparison of previously known algorithms and the results of this paper. The algorithms for cutwidth and pathwidth take as input a semi-complete digraph $T$ on $n$ vertices and an integer $k$. The approximation algorithms can output a decomposition of larger width (the guarantees are mentioned in the corresponding cells) or correctly conclude that constructing decomposition of width at most $k$ is impossible. The exact algorithms either construct a decomposition of width at most $k$, or correctly conclude that this is impossible.}
\end{centering}
\end{figure}

The crucial observation of the new approach can be intuitively formulated as follows. Consider a semi-complete digraph with a large number of vertices, but admitting a path decomposition of very small width. Take any vertex $v$ that appears in a small number of bags. Observe that for each vertex $w$ that is located much further in the decomposition, the arc between $v$ and $w$ must be directed from $w$ to $v$. Similarly, if $w$ is much earlier, then the arc must be directed from $v$ to $w$. Hence, the outdegree of $v$ is more or less equal to the number of vertices that appear before $v$ in the decomposition; the only aberrations that can occur are due to vertices that interact with $v$ within the bags, but their number is limited by the width of the decomposition. Therefore, intuitively any ordering of the vertices with respect to increasing outdegrees should be a good approximation of the order in which the vertices appear in the optimal path decomposition. 

In fact, for cutwidth this is true even in formal sense. We prove that any outdegree ordering of vertices of a semi-complete digraph $T$ has width at most $O(\ctw(T)^2)$, hence we have a trivial approximation algorithm that sorts the vertices with respect to outdegrees. The exact algorithm for cutwidth is based on an observation that an optimal ordering and any outdegree ordering can differ only locally: if $X$ is the set of $\ell$ first vertices in the optimal ordering, then $X$ contains first $\ell-O(k)$ vertices of the degree ordering, and is contained in first $\ell+O(k)$ vertices of the outdegree ordering. Hence, we scan through the outdegree ordering with a dynamic program, maintaining a bit mask of length $O(k)$ denoting which vertices of an appropriate interval are contained in constructed prefix of an optimal ordering. 

The case of pathwidth, which is of our main interest, is more complicated. We formalize the intuitive outdegree ordering argument as follows: we prove that existence of $5k+2$ vertices with outdegrees mutually not differing by more than $k$ already forms an obstacle for admitting a path decomposition of width at most $k$; we call this obstacle a {\em{degree tangle}}. Hence, any outdegree ordering of vertices of a given semi-complete digraph $T$ of small pathwidth must be already quite spread: it does not contain larger clusters of vertices with similar outdegrees. This spread argument is crucial in all our reasonings, and shows that finding new structural obstacles can significantly simplify our view of the problem.

Both the approximation and exact algorithm for pathwidth use the concept of scanning through the outdegree ordering with a {\em{window}} --- an interval in the ordering containing $5k$ vertices. By the outdegree spread argument, at each point we know that the vertices on the left side of the window have outdegrees smaller by more than $k$ that the ones on the right side; otherwise we would have a degree tangle. For approximation, we construct the consecutive bags by greedily taking the window and augmenting this choice with a small coverage of arcs jumping over it. The big gap between outdegrees on the left and on the right side of the window ensures that nonexistence of a small coverage is also an evidence for not admitting a decomposition of small width. For exact algorithm, we identify a set of $O(k^2)$ vertices around the window, about which we can safely assume that the bag is contained in it. Then we run a similar dynamic programming algorithm as in case of cutwidth.

The most technical part in the approximation and exact algorithms for pathwidth is the choice of vertices outside the window to cover the arcs jumping over it. It turns out that this problem can be expressed as trying to find a small vertex cover in an auxiliary bipartite graph. However, in order to obtain a feasible path decomposition we cannot choose the vertex cover arbitrarily --- it must behave consistently as the window slides through the ordering. To this end, in the approximation algorithm we use a $2$-approximation of the vertex cover based on the matching theory. In the exact algorithm we need more restrictions, as we seek a subset that contains {\em{every}} sensible choice of the bag. Therefore, we use an $O(OPT)$-approximation of vertex cover based on the classical Buss kernelization routine for the problem~\cite{buss}, which enables us to use stronger arguments to reason which vertices can be excluded from consideration.

Finally, we observe that the obtained set of obstacles for pathwidth have much better properties than the ones introduced by Fradkin and Seymour~\cite{fradkin-seymour}. We show that there is a constant multiplicative factor relation between the sizes of obstacles found by the algorithm and the minimum size of a digraph that cannot be embedded into them. Hence, in order to test if $H$ is topologically contained in a given semi-complete digraph $T$, we just need to run the pathwidth approximation with parameter $O(|H|)$, and in case of finding an obstacle just provide a positive answer. This trims the dependency of running time of the topological subgraph and immersion tests to single exponential in terms of the size of tested subgraph, compared to multiple-exponential as presented in~\cite{my,ChudnovskyFS2011}.

\vskip 0.15cm

\noindent{\bf{Organization.}} In Section~\ref{sec:prel} we introduce notation and main definitions. In Section~\ref{sec:zoo} we present the obstacles and prove their basic properties. In Section~\ref{sec:cutwidth} we deal with cutwidth and in Section~\ref{sec:pathwidth} with pathwidth. In Section~\ref{sec:top} we apply the introduced tools to the containment problems, while Section~\ref{sec:conclusions} gathers concluding remarks. Appendix~\ref{app:dp} contains the description of a dynamic programming routine for topological containment on semi-complete digraphs with a path decomposition of small width; the routine is very similar to the one for immersion, described in~\cite{my}, but we include it for the sake of completeness.
 
\section{Preliminaries}\label{sec:prel}

\subsection{Notation and basic definitions}

We use standard graph notation. If $G$ is a (di)graph, by $V(G)$ we denote the vertex set of $G$ and by $E(G)$ the edge (arc) set of $G$; we define size of $G$ to be equal to $|G|=|V(G)|+|E(G)|$. By $G[X]$ we denote the sub(di)graph induced by $X$. For $X,Y\subseteq V(G)$ we define $E(X,Y)=\{(v,w)\in E(G)\ |\ v\in V, w\in W\}$. If $G$ is undirected, then $N(v)$ denotes the set of neighbours of $v$; for directed graphs, by $N^+(v),N^-(v)$ we denote the sets of outneighbours and inneighbours of $v$, respectively. We extend this notion to subsets naturally, e.g., $N(X)=\bigcup_{v\in X} N(v) \setminus X$. We define the {\em{outdegree}} and {\em{indegree}} of $v$ as $d^+(v)=|N^+(v)|$ and $d^-(v)=|N^-(v)|$, respectively. An {\em{outdegree ordering}} is an ordering of vertices with respect to nondecreasing outdegrees.

A digraph is {\em{simple}}, if it does not contain loops or multiple arcs; all the digraphs considered in this paper are simple. A digraph is {\em{semi-complete}}, if it is simple and for every two distinct vertices $v,w$ {\bf{at least}} one of the arcs $(v,w)$ or $(w,v)$ is present. Note that we do not exclude existence of both of these arcs, but we exclude existence of, for instance, two arcs $(v,w)$. A semi-complete digraph is a {\em{tournament}}, if for every two distinct vertices $v,w$ {\bf{exactly}} one of the arcs $(v,w)$ or $(w,v)$ is present.

A {\em{separation}} in a digraph $T$ is a pair $(A,B)$ such that $A,B\subseteq V(T)$, $A\cup B=V(T)$ and $E(A\setminus B,B\setminus A)=\emptyset$. The {\em{order}} of a separation is the size of the {\em{separator}}, i.e., $|A\cap B|$.

\subsection{Containment notions for digraphs}

{\em{Immersion}} and {\em{topological containment}} are the two main notions of containment that are under consideration in this paper. They are direct analogues of the classical undirected versions; for the sake of completeness, we include here their formal definitions.

We follow notation from~\cite{my}. Let $H,G$ be digraphs. We say that mapping $\eta : V(H)\cup E(H) \to V(G)\cup E(G)$ is a {\em{model}} of $H$ in $G$, if the following conditions are satisfied:
\begin{itemize}
\item $\eta(V(H))\subseteq V(G)$;
\item for every arc $(u,v)\in E(H)$, $\eta((u,v))$ is a directed path leading from $\eta(u)$ to $\eta(v)$.
\end{itemize}
By imposing further conditions on the model we obtain various notions of containment for digraphs. If we require that $\eta$ is surjective on $V(H)$ and the paths in $\eta(E(H))$ are internally vertex-disjoint, we obtain the notion of {\emph{topological containment}}; in this case we say that $\eta$ is an {\emph{expansion}} of $H$ in $G$. If we relax the condition of paths being internally vertex-disjoint to being edge-disjoint, we arrive at the notion of {\emph{immersion}}; then $\eta$ is an {\emph{immersion}} of $H$ into~$G$. 

Sometimes in the literature this notion is called {\emph{weak immersion}} to distinguish it from {\emph{strong immersion}}, where each image of an arc is additionally required not to pass through images of vertices not being the endpoints of this arc. In this work we are interested only in (weak) immersions, as we find the notion of strong immersion not well-motivated enough to investigate all the technical details that arise when considering both definitions at the same time and discussing slight differences between them.  

We say that $\mathbf{G}=(G;v_1, \dots, v_h)$ is a \emph{rooted} digraph if $G$ is digraph and $v_1, \dots, v_h$ are distinct vertices of $V(G)$. The notions of immersion and topological containment can be naturally generalized to rooted digraphs. Immersion $\eta$ is an immersion from rooted digraph  $\mathbf{H}=(H;u_1, \dots, u_h)$ to rooted digraph  $\mathbf{G}=(G;v_1, \dots, v_h)$ if $\eta(u_i)=v_i$ for $i\in \{1,\dots, h\}$. Such an immersion is called an \emph{$\mathbf{H}$-immersion} or a \emph{rooted immersion}. In the same manner we may define \emph{$\mathbf{H}$-expansions} or \emph{rooted expansions}.

\subsection{Width parameters of digraphs}\label{sec:width}

\begin{definition}\label{def:ctw}
Given a digraph $G=(V,E)$ and an ordering $\pi$ of $V$, let $\pi[\alpha]$ be the first $\alpha$ vertices in the ordering $\pi$. The \emph{width} of $\pi$ is equal to $\max_{0\leq \alpha\leq |V|} |E(\pi[\alpha],V\setminus \pi[\alpha])|$; the \emph{cutwidth} of $G$, denoted $\ctw(G)$, is the minimum width over all orderings of $V$.
\end{definition}

\begin{definition}\label{def:pw}
Given a digraph $G=(V,E)$, a sequence $W = (W_1, \dots , W_r)$ of subsets of $V$  is a \emph{path decomposition of $G$} if the following conditions are satisfied:
 \begin{itemize}
\item[{\em{(i)}}] $\bigcup_{1\leq i\leq r} W_i = V $;
\item[{\em{(ii)}}] $W_i \cap W_k \subseteq W_j$ for $1 \leq i < j < k \leq r$;
\item[{\em{(iii)}}]  $\forall$  $(u,v)\in E$, either $u,v\in W_i$ for some $i$ or $u\in W_i$, $v\in W_j$ for some $i>j$.  \end{itemize}
 We call $W_1,\dots , W_r$ the \emph{bags} of the path decomposition.
 The \emph{width} of a path decomposition is equal to
$\max_{1\leq i \leq r} (|W_i| - 1)$; the \emph{pathwidth} of $G$, denoted $\pw(G)$, is the minimum width over all path decompositions of $G$.
\end{definition}

In~\cite{my} it is observed that $\pw(G)\leq 2\ctw(G)$ for every digraph $G$. We say that a path decomposition $W=(W_1,\ldots,W_{r})$ is {\em{nice}} if it has following two additional properties:
\begin{itemize}
\item $W_1=W_{r}=\emptyset$;
\item for every $i=1,2,\ldots,r-1$ we have that $W_{i+1}=W_i\cup \{v\}$ for some vertex $v\notin W_i$, or $W_{i+1}=W_i\setminus \{w\}$ for some vertex $w\in W_i$.
\end{itemize}
If $W_{i+1}=W_i\cup \{v\}$ then we say that in bag $W_{i+1}$ we {\em{introduce}} vertex $v$, while if $W_{i+1}=W_i\setminus \{w\}$ then we say that in bag $W_{i+1}$ we {\em{forget}} vertex $w$. Given any path decomposition $W$ of width $p$, in $O(p|V(T)|)$ time we can construct a nice path decomposition $W'$ of the same width in a standard manner: we first introduce empty bags at the beginning and at the end, and then insert new bags between any two consecutive ones by firstly forgetting all the vertices that do not appear in the second bag, and then introducing all the new ones.

Any path decomposition naturally corresponds to a monotonic sequence of separations.

\begin{definition}
A sequence of separations  $((A_0,B_0),\ldots,(A_r,B_r))$ is called a {\em{separation chain}} if $(A_0,B_0)=(\emptyset,V(T))$, $(A_r,B_r)=(V(T),\emptyset)$ and $A_i\subseteq A_j, B_i\supseteq B_j$ for all $i\leq j$. The {\em{width}} of the separation chain is equal to $\max_{0\leq i\leq r} |A_i\cap B_i|$.
\end{definition}

\begin{lemma}\label{lem:sepchain}
The following holds.
\begin{itemize}
\item Let $W=(W_1,\ldots,W_r)$ be a path decomposition of a digraph $T$ of width at most $p$, where no two consecutive bags are equal. Then sequence $((A_0,B_0),\ldots,(A_r,B_r))$ defined as $(A_i,B_i)=(\bigcup_{j=1}^i W_j, \bigcup_{j=i+1}^r W_j)$ is a separation chain of width at most $p$.
\item Let $((A_0,B_0),\ldots,(A_r,B_r))$ be a separation chain in digraph $T$. Then $W=(W_1,\ldots,W_r)$ defined by $W_i=A_i\cap B_{i-1}$ is a path decomposition of width $\max_{1\leq i\leq r} |A_i\cap B_{i-1}|-1$.
\end{itemize}
\end{lemma}
\begin{proof}
For the first claim, it suffices to observe that $(A_i,B_i)$ are indeed separations, as otherwise there would be an edge $(v,w)\in E(T)$ such that $v$ can belong only to bags with indices at most $i$ and $w$ can belong only to bags with indices larger than $i$; this is a contradiction with property {\em{(iii)}} of path decomposition. The bound on width follows from the fact that $A_i\cap B_i=W_i\cap W_{i+1}$ (by property {\em{(ii)}} of path decomposition) and no two consecutive bags are equal.

For the second claim, we need to carefully check all the properties of a path decomposition. Property {\em{(i)}} follows from the fact that  $A_0=\emptyset$, $A_r=V(T)$ and $W_i\supseteq A_i\setminus A_{i-1}$  for all $1\leq i\leq r$. Property {\em{(ii)}} follows from the fact that $A_i\subseteq A_j, B_i\supseteq B_j$ for all $i\leq j$: the interval in the decomposition containing any vertex $v$ corresponds to the intersection of the prefix of the chain where $v$ belongs to sets $A_i$, and the suffix where $v$ belongs to sets $B_i$.

For property {\em{(iii)}}, take any $(v,w)\in E(T)$. Let $\alpha$ be the largest index such that $v\in W_\alpha$ and $\beta$ be the smallest index such that $w\in W_\beta$. It suffices to prove that $\alpha\geq \beta$. For the sake of contradiction assume that $\alpha<\beta$ and consider separation $(A_\alpha,B_\alpha)$. By maximality of $\alpha$ it follows that $v\notin B_{\alpha}$; as $\beta>\alpha$ and $\beta$ is minimal, we have also that $w\notin A_{\alpha}$. Then $v\in A_{\alpha}\setminus B_{\alpha}$ and $w\in B_{\alpha}\setminus A_{\alpha}$, which contradicts the fact that $(A_{\alpha},B_{\alpha})$ is a separation.
\end{proof}

Assume that $((A_0,B_0),\ldots,(A_r,B_r))$ is a separation chain corresponding to a nice path decomposition $W$ in the sense of Lemma~\ref{lem:sepchain}. Since $A_0=\emptyset$, $A_r=V(G)$, and $|A_i|$ can change by at most $1$ between two consecutive separations, for every $\ell$, $0\leq \ell\leq |V(G)|$, there is some separation $(A_i,B_i)$ for which $|A_i|=\ell$ holds. Let $W[\ell]$ denote any such separation.

\subsection{Degrees in semi-complete digraphs}

Let $T$ be a semi-complete digraph. Note that for every $v\in V(T)$ we have that $d^+(v)+d^-(v)\geq |V(T)|-1$, and that the equality holds for all $v\in V(T)$ if and only if $T$ is a tournament.

\begin{lemma}\label{lem:small-outdeg}
Let $T$ be a semi-complete digraph. Then the number of vertices of $T$ with outdegrees at most $d$ is at most $2d+1$.
\end{lemma}
\begin{proof}
Let $A$ be the set of vertices of $T$ with outdegrees at most $d$, and for the sake of contradiction assume that $|A|>2d+1$. Consider semi-complete digraph $T[A]$. As in every semi-complete digraph $S$ there is a vertex of outdegree at least $\frac{|V(S)|-1}{2}$, in $T[A]$ there is a vertex with outdegree larger than $d$. As outdegrees in $T[A]$ are not smaller than in $T$, this is a contradiction with the definition of $A$.
\end{proof}

\begin{lemma}\label{lem:degree-gap}
Let $T$ be a semi-complete digraph and let $x,y$ be vertices of $T$ such that $d^+(x)>d^+(y)+\ell$. Then there exist at least $\ell$ vertices that are both outneighbours of $x$ and inneighbours of $y$ and, consequently, $\ell$ vertex-disjoint paths of length $2$ from $x$ to $y$.
\end{lemma}
\begin{proof}
Let $\alpha=d^+(y)$. We have that $d^-(y)+d^+(x)\geq |V(T)|-1-\alpha+\alpha+\ell+1=|V(T)|+\ell$. Hence, by the pigeonhole principle there exist at least $\ell$ vertices of $T$ that are both outneighbours of $x$ and inneighbours of $y$.
\end{proof}

\section{The obstacle zoo}\label{sec:zoo}

In this section we describe the set of obstacles used by the algorithms. We begin with {\em{short jungles}}, enhanced analogues of the $k$-jungle of Fradkin and Seymour~\cite{fradkin-seymour}. It appears that the enhancement enables us to construct large topological subgraph or immersion models in short jungles in a greedy manner, and this observation is the key to trimming the running times for topological containment and immersion tests. We continue with further obstacles: degree and matching tangles for pathwidth, and backward tangles for cutwidth, each time proving two lemmas. The first asserts that existence of the structure is indeed an obstacle for having small width, while the second shows that one can constructively find an appropriate short jungle in a sufficiently large obstacle.

\subsection{Short jungles}

\begin{definition}
Let $T$ be a semi-complete digraph and $k,d$ be integers. A $(k,d)${\em{-short (immersion) jungle}} is a set $X\subseteq V(T)$ such that {\em{(i)}} $|X|\geq k$; {\em{(ii)}} for every $v,w\in X$ there are $k$ vertex-disjoint (edge-disjoint) paths from $v$ to $w$ of length at most $d$. 
\end{definition}

We remark that in all our algorithms, every short jungle is constructed and stored together with corresponding families of $k$ paths for each pair of vertices. The restriction on the length of the paths enables us to construct topological subgraph (immersion) models greedily.

\begin{lemma}\label{lem:sj-ts}
If a digraph $T$ contains a $(dk,d)$-short (immersion) jungle for some $d>1$, then it admits every digraph $S$ with $|S|\leq k$ as a topological subgraph (as an immersion).
\end{lemma}
\begin{proof}
Firstly, we prove the lemma for short jungles and topological subgraphs. Let $X$ be the short jungle whose existence is assumed. We construct the expansion model greedily. As images of vertices of $S$ we put arbitrary $|V(S)|$ vertices of $X$. Then we construct paths being images of arcs in $S$; during each construction we use at most $d-1$ new vertices of the graph for the image. While constructing the $i$-th path, which has to lead from $v$ to $w$, we consider $dk$ vertex-disjoint paths of length $d$ from $v$ to $w$. So far we used at most $k+(i-1)(d-1)<dk$ vertices, so at least one of these paths does not traverse any used vertex. Hence, we can safely use this path as the image and proceed; note that thus we use at most $d-1$ new vertices.

Secondly, we prove the lemma for short immersion jungles and immersions. Let $X$ be the short immersion jungle whose existence is assumed. We construct the immersion model greedily. As images of vertices of $S$ we put arbitrary $|V(S)|$ vertices of $X$. Then we construct paths being images of arcs in $S$; during each construction we use at most $d$ new arcs of the graph for the image. While constructing the $i$-th path, which has to lead from $v$ to $w$, we consider $dk$ edge-disjoint paths of length $d$ from $v$ to $w$. So far we used at most $(i-1)d<dk$ arcs, so at least one of these paths does not contain any used arc. Hence, we can safely use this path as the image and proceed; note that thus we use at most $d$ new arcs.
\end{proof}

\subsection{Degree tangles}

\begin{definition}
Let $T$ be a semi-complete digraph and $k,\ell$ be integers. A $(k,\ell)${\em{-degree tangle}} is a set $X\subseteq V(T)$ such that {\em{(i)}} $|X|\geq k$; {\em{(ii)}} for every $v,w\in X$ we have $|d^+(v)-d^+(w)|\leq \ell$.
\end{definition}

\begin{lemma}\label{lem:degree-tangle-pathwidth}
Let $T$ be a semi-complete digraph. If $T$ contains a $(5k+2,k)$-degree tangle~$X$, then $\pw(T)>k$.
\end{lemma}
\begin{proof}
For the sake of contradiction, assume that $T$ admits a (nice) path decomposition $W$ of width at most $k$. Let $\alpha=\min_{v\in X} d^+(v)$ and $\beta=\max_{v\in X} d^+(v)$; we know that $\beta-\alpha\leq k$. Let $(A,B)=W[\alpha]$. We know that $|A\cap B|\leq k$ and $|A|=\alpha$.

Firstly, observe that $X\cap (A\setminus B)=\emptyset$. This follows from the fact that vertices in $A\setminus B$ can have outneighbours only in $A$, so their outdegrees are upper bounded by $|A|-1=\alpha-1$.

Secondly, $|X\cap (A\cap B)|\leq k$ as $|A\cap B|\leq k$.

Thirdly, we claim that $|X\cap (B\setminus A)|\leq 4k+1$. Assume otherwise. By Lemma~\ref{lem:small-outdeg} we infer that there exists a vertex $v\in X\cap (B\setminus A)$ whose outdegree in $T[B\setminus A]$ is at least $2k+1$. As $(A,B)$ is a separation and $T$ is semi-complete, all the vertices of $A\setminus B$ are also outneighbours of $v$. Note that $|A\setminus B|=|A|-|A\cap B|\geq \alpha-k$. We infer that $v$ has at least $\alpha-k+2k+1=\alpha+k+1>\beta$ neighbours, which is a contradiction with $v\in X$.

Summing up the bounds we get $5k+2\leq |X|\leq k+4k+1=5k+1$, a contradiction.
\end{proof}

\begin{lemma}\label{lem:degree-tangle-jungle}
Let $T$ be a semi-complete digraph and let $X$ be a $(26k,k)$-degree tangle in $T$. Then $X$ contains a $(k,3)$-short jungle, which can be found in $O(k^3 |V(T)|^2)$ time.
\end{lemma}
\begin{proof}
We present the proof of the existence statement. The proof can be easily turned into an algorithm finding the jungle; during the description we make remarks at the places where it may be non-trivial to observe how the algorithm should perform to achieve the promised running-time guarantee.

By possibly trimming $X$, assume without loss of generality that $|X|=26k$. Take any $v,w\in X$. We are going to either to find a $(k,3)$-short jungle in $X$ explicitly, or find $k$ vertex-disjoint paths from $v$ to $w$ of length at most $3$. If for no pair an explicit short jungle is found, we conclude that $X$ is a $(k,3)$-short jungle itself.

Let us consider four subsets of $V(T)\setminus \{v,w\}$: 
\begin{itemize}
\item $V^{++}=(N^+(v)\cap N^+(w))\setminus \{v,w\}$, 
\item $V^{+-}=(N^+(v)\cap N^-(w))\setminus \{v,w\}$, 
\item $V^{-+}=(N^-(v)\cap N^+(w))\setminus \{v,w\}$,
\item $V^{--}=(N^-(v)\cap N^-(w))\setminus \{v,w\}$.
\end{itemize}
Clearly, $|V^{++}|+|V^{-+}|+|V^{+-}|+|V^{--}|\geq |V(T)|-2$. Note that equality holds for the tournament case --- in this situation these four subsets form a partition of $V(T)\setminus \{v,w\}$.

If $|V^{+-}|\geq k$, then we already have $k$ vertex-disjoint paths of length $2$ from $v$ to $w$. Assume then that $|V^{+-}|<k$.

Observe that $d^+(v)\leq |V^{++}|+|V^{+-}|+1$ and $d^+(w)\geq |V^{++}|+|V^{-+}\setminus V^{++}|$. Since $v,w\in X$, we have that $d^+(w)-d^+(v)\leq k$, so
$$|V^{-+}\setminus V^{++}|\leq d^+(w)-|V^{++}|\leq k+d^+(v)-|V^{++}|\leq k+1+|V^{+-}|\leq 2k.$$
Let $A=V^{++}\setminus V^{+-}$ and $B=V^{--}\setminus V^{+-}$. Note that $A$ and $B$ are disjoint, since $V^{++}\cap V^{--}\subseteq V^{+-}$. Let $H$ be a bipartite graph with bipartition $(A,B)$, such that for $a\in A$ and $b\in B$ we have $ab\in E(H)$ if and only if $(a,b)\in E(T)$. Every edge $ab$ of $H$ gives raise to a path $v\to a\to b\to w$ of length $3$ from $v$ to $w$. Hence, if we could find a matching of size $k$ in $H$, then this matching would form a family of $k$ vertex disjoint paths of length at most $3$ from $v$ to $w$. Note that testing existence of such a matching can be done in $O(k|V(T)|^2)$ time, as we can run the algorithm finding an augmenting path at most $k$ times.

Assume then that such a matching does not exist. By K\"{o}nig's theorem we can find a vertex cover $C$ of $H$ of cardinality smaller than $k$; again, this can be found in $O(k|V(T)|^2)$ time. As $A\cup B\cup (V^{-+}\setminus V^{++})\cup V^{+-}=V(T)\setminus \{v,w\}$ while $(V^{-+}\setminus V^{++})\cup V^{+-}$ contains at most than $3k-1$ vertices in total, $A\cup B$ must contain at least $(26k-2)-(3k-1)=23k-1$ vertices from $X$. We consider two cases: either $|A\cap X|\geq 16k$, or $|B\cap X|\geq 7k$.

\vskip 0.2cm

\noindent{\bf{Case 1.}} In the first case, consider set $Y_0=X\cap (A\setminus C)$. Since $|A\cap X|\geq 16k$ and $|A\cap C|<k$, we have that $|Y_0|>15k$. Let $Y$ be any subset of $Y_0$ of size $15k$. Take any vertex $y\in Y$ and consider, where its outneighbours can lie. These outneighbours can be either in $\{v\}$ (at most $1$ of them), in $(V^{-+}\setminus V^{++})\cup V^{+-}$ (less than $3k$ of them), in $B\cap C$ (at most $k$ of them), or in $A$. As $d^+(v)\geq |A|$ and $v,y\in X$, we have that $d^+(y)\geq|A|-k$. We infer that $y$ must have at least $|A|-5k$ outneighbours in $A$. As $|Y|=15k$, we have that $y$ has at least $10k$ outneighbours in $Y$.

Note that in the tournament case we would be already finished, as this lower bound on the outdegree would imply also an upper bound on indegree, which would contradict the fact that $T[Y]$ contains a vertex of indegree at least $\frac{|Y|-1}{2}$. This also shows that in the tournament setting a stronger claim holds that in fact $X$ is a $(k,3)$-jungle itself. In the semi-complete setting, however, we do not have any contradiction yet. In fact no contradiction is possible as the stronger claim is no longer true. To circumvent this problem, we show how to find an explicit $(k,3)$-jungle within $Y$.

Observe that the sum of outdegrees in $T[Y]$ is at least $10k\cdot 15k=150k^2$. We claim that the number of vertices in $Y$ that have indegrees at least $6k$ is at least $k$. Otherwise, the sum of indegrees would be bounded by $15k\cdot k + 6k\cdot 14k=99k^2<150k^2$ and the sums of the indegrees and of the outdegrees would not be equal. Let $Z$ be any set of $k$ vertices in $Y$ that have indegrees at least $6k$ in $T[Y]$. Take any $z_1,z_2\in Z$ and observe that in $T[Y]$ the set of outneighbours of $z_1$ and the set of inneighbours of $z_2$ must have intersection of size at least $k$, as $d^+_{T[Y]}(z_1)\geq 10k$, $d^-_{T[Y]}(z_2)\geq 6k$ and $|Y|=15k$. Through these $k$ vertices one can pass $k$ vertex-disjoint paths from $z_1$ to $z_2$, each of length $2$. Hence, $Z$ is the desired $(k,3)$-short jungle.

\vskip 0.2cm

\noindent{\bf{Case 2.}} This case will be similar to the previous one, with the exception that we only get a contradiction: there is no subcase with finding an explicit smaller jungle. Consider set $Y_0=X\cap (B\setminus C)$. Since $|B\cap X|\geq 7k$ and $|B\cap C|<k$, we have that $|Y_0|>6k$. Let $Y$ be any subset of $Y_0$ of size $6k+1$. Take any vertex $y\in Y$ that has outdegree at least $3k$ in $T[Y]$ (since $|Y|=6k+1$, such a vertex exists), and consider its outneighbours. As $y\notin C$ we have that all the vertices of $A\setminus C$ are the outneighbours of $y$ (more than $|A|-k$ of them), and there are at least $3k$ outneighbours within $B$. Hence $d^+(y)>|A|+2k$. On the other hand, the outneighbours of $v$ have to lie inside $A\cup V^{+-}\cup \{w\}$, so $d^+(v)\leq |A\cup V^{+-}\cup \{w\}|\leq |A|+k$. We infer that $d^+(y)-d^+(v)>k$, which is a contradiction with $v,y\in X$.
\end{proof}

\subsection{Matching tangles}

\begin{definition}
Let $T$ be a semi-complete digraph and $k,\ell$ be integers. A $(k,\ell)${\em{-matching tangle}} is a pair of disjoint subsets $X,Y\subseteq V(T)$ such that {\em{(i)}} $|X|=|Y|=k$; {\em{(ii)}} there exists a matching from $X$ to $Y$, i.e., there is a bijection $f:X\to Y$ such that $(v,f(v))\in E(T)$ for all $v\in X$; {\em{(iii)}} for every $v\in X$ and $w\in Y$ we have that $d^+(w)>d^+(v)+\ell$.
\end{definition}

\begin{lemma}\label{lem:matching-tangle-pathwidth}
Let $T$ be a semi-complete digraph. If $T$ contains a $(k+1,k)$-matching tangle $(X,Y)$, then $\pw(T)>k$.
\end{lemma}
\begin{proof}
For the sake of contradiction assume that $T$ has a (nice) path decomposition $W$ of width at most $k$. Let $\alpha=\min_{w\in Y} d^+(w)$ and let $(A,B)=W[\alpha]$. Recall that $|A\cap B|\leq k$.

Firstly, we claim that $X\subseteq A$. Assume otherwise that there exists some $v\in (B\setminus A)\cap X$. Note that all the vertices of $A\setminus B$ are outneighbours of $v$, so $d^+(v)\geq |A|-k=\alpha-k$. Hence $d^+(v)\geq d^+(w)-k$ for some $w\in Y$, which is a contradiction.

Secondly, we claim that $Y\subseteq B$. Assume otherwise that there exists some $w\in (A\setminus B)\cap Y$. Then all the outneighbours of $w$ must lie within $A$, so there is less than $\alpha$ of them. This is a contradiction with the definition of $\alpha$.

As $|A\cap B|\leq k$ and there are $k+1$ disjoint pairs of form $(v,f(v))\in E(T)$ for $v\in X$, we conclude that there must be some $v\in X$ such that $v\in A\setminus B$ and $f(v)\in B\setminus A$. This contradicts the fact that $(A,B)$ is a separation.
\end{proof}

\begin{lemma}\label{lem:matching-tangle-jungle}
Let $T$ be a semi-complete digraph and let $(X,Y)$ be a $(5k,3k)$-matching tangle in $T$. Then $Y$ contains a $(k,4)$-short jungle, which can be found in $O(k^3|V(T)|)$ time.
\end{lemma}
\begin{proof}[Proof of Lemma~\ref{lem:matching-tangle-jungle}]
We present the proof of the existential statement; all the steps of the proof are easily constructive and can be performed within the claimed complexity bound.

Let $Z$ be the set of vertices of $Y$ that have indegrees at least $k+1$ in $T[Y]$. We claim that $|Z|\geq k$. Otherwise, the sum of indegrees in $T[Y]$ would be at most $k\cdot 5k+4k\cdot k=9k^2<\binom{5k}{2}$, so the total sum of indegrees would be strictly smaller than the number of arcs in the digraph. It remains to prove that $Z$ is a $(k,4)$-short jungle.

Take any $v,w\in Z$; we are to construct $k$ vertex-disjoint paths from $v$ to $w$, each of length at most $4$. Let $R_0=N^{-}_{T[Y]}(w)\setminus \{v\}$. Note that $|R_0|\geq k$, hence let $R$ be any subset of $R_0$ of cardinality $k$ and let $P=f^{-1}(R)$. We are to construct $k$ vertex-disjoint paths of length $2$ connecting $v$ with every vertex of $P$ and not passing through $P\cup R\cup \{w\}$. By concatenating these paths with arcs of the matching $f$ between $P$ and $R$ and arcs leading from $R$ to $w$, we obtain the family of paths we look for.

The paths from $v$ to $P$ are constructed in a greedy manner, one by one. Each path construction uses exactly one vertex outside $P\cup R$. Let us take the next, $i$-th vertex $p\in P$. As $d^+(v)>d^+(p)+3k$, by Lemma~\ref{lem:degree-gap} there exist at least $3k$ vertices in $T$ that are both outneighbours of $v$ and inneighbours of $p$. At most $2k$ of them can be inside $P\cup R$, at most $i-1\leq k-1$ of them were used for previous paths, so there is at least one that is still unused; let us denote it by $q$. If in fact $q=w$, we build a path of length $1$ directly from $v$ to $w$ thus ignoring vertex $p$; otherwise we can build the path of length $2$ from $v$ to $p$ via $q$ and proceed to the next vertex of $P$.
\end{proof}

\subsection{Backward tangles}

\begin{definition}
Let $T$ be a semi-complete digraph and $k$ be an integer. A $k${\em{-backward tangle}} is a partition $(X,Y)$ of $V(T)$ such that {\em{(i)}} there exist at least $k$ arcs directed from $X$ to $Y$; {\em{(ii)}} for every $v\in X$ and $w\in Y$ we have that $d^+(w)\geq d^+(v)$.
\end{definition}

\begin{lemma}\label{lem:backward-tangle-cutwidth}
Let $T$ be a semi-complete digraph. If $T$ contains an $(m+1)$-backward tangle $(X,Y)$ for $m=100k^2+22k+1$, then $\ctw(T)>k$.
\end{lemma}
\begin{proof}
For the sake of contradiction, assume that $V(T)$ admits an ordering $\pi$ of width at most $k$. Let $\alpha$ be the largest index such that $(X_\alpha,Y_\alpha)=(\pi[\alpha],V(T)\setminus \pi[\alpha])$ satisfies $Y\subseteq Y_\alpha$. Similarly, let $\beta$ be the smallest index such that $(X_\beta,Y_\beta)=(\pi[\beta],V(T)\setminus \pi[\beta])$ satisfies $X\subseteq X_\beta$. Note that $|E(X_\alpha,Y_\alpha)|,|E(X_\beta,Y_\beta)|\leq k$. Observe also that $\alpha\leq \beta$; moreover, $\alpha<|V(T)|$ and $\beta>0$, since $X,Y$ are non-empty.

Let $(X_{\alpha+1},Y_{\alpha+1})=(\pi[\alpha+1],V(T)\setminus \pi[\alpha+1])$. By the definition of $\alpha$ there is a unique vertex $w\in X_{\alpha+1}\cap Y$. Take any vertex $v\in V(T)$ and suppose that $d^+(w)>d^+(v)+(k+1)$. By Lemma~\ref{lem:degree-gap}, there exist $k+1$ vertex-disjoint paths of length $2$ from $w$ to $v$. If $v$ was in $Y_{\alpha+1}$, then each of these paths would contribute at least one arc to the set $E(X_{\alpha+1},Y_{\alpha+1})$, contradicting the fact that $|E(X_{\alpha+1},Y_{\alpha+1})|\leq k$. Hence, every such $v$ belongs to $X_{\alpha+1}$ as well. By Lemma~\ref{lem:degree-tangle-pathwidth} we have that the number of vertices with outdegrees in the interval $[d^+(w)-(k+1),d^+(w)]$ is bounded by $10k+1$, as otherwise they would create a $(10k+2,2k)$-degree tangle, implying that $\pw(T)>2k$ and, consequently, $\ctw(T)>k$ (here note that for $k=0$ the lemma is trivial). As $X_{\alpha}=X_{\alpha+1}\setminus \{w\}$ is disjoint with $Y$ and all the vertices of $X$ have degrees at most $d^+(w)$, we infer that $|X\setminus X_\alpha|\leq 10k+1$.

A symmetrical reasoning shows that $|Y\setminus Y_\beta|\leq 10k+1$. Now observe that
\begin{eqnarray*}
|E(X,Y)| & \leq & |E(X_\alpha,Y)|+|E(X,Y_\beta)|+|E(X\setminus X_\alpha,Y\setminus Y_\beta)| \\
         & \leq & |E(X_\alpha,Y_\alpha)|+|E(X_\beta,Y_\beta)|+|E(X\setminus X_\alpha,Y\setminus Y_\beta)| \\
         & \leq & k+k+(10k+1)^2 = 100k^2+22k+1.
\end{eqnarray*}
This is a contradiction with $(X,Y)$ being an $(m+1)$-backward tangle.
\end{proof}

\begin{lemma}\label{lem:backward-tangle-jungle}
Let $T$ be a semi-complete digraph and let $(X,Y)$ be an $m$-backward tangle in $T$ for $m=109^2k$. Then $X$ or $Y$ contains a $(k,4)$-short immersion jungle, which can be found in $O(k^3 |V(T)|^2)$ time.
\end{lemma}
\begin{proof}
We present the proof of the existential statement; all the steps of the proof are easily constructive and can be performed within the claimed complexity bound.

Let $\alpha$ be the maximum outdegree in $X$. Let $P_0\subseteq X$ and $Q_0\subseteq Y$ be the sets of heads and of tails of arcs from $E(X,Y)$, respectively. As $|E(X,Y)|\leq |P_0|\cdot |Q_0|$, we infer that $|P_0|\geq 109k$ or $|Q_0|\geq 109k$. Here we consider the first case; the reasoning in the second one is symmetrical.

Let $P_1$ be the set of vertices in $P_0$ that have outdegree at least $\alpha-4k$, where $\alpha=\min_{w\in Y}d^+(w)$; note that $\alpha\geq \max_{v\in X}d^+(v)$ by the definition of a backward tangle. If there were more than $104k$ of them, they would create a $(104k,4k)$-degree tangle, which due to Lemma~\ref{lem:degree-tangle-jungle} contains a $(4k,3)$-short jungle, which is also a $(k,4)$-short immersion jungle. Hence, we can assume that $|P_1|<104k$. Let $P$ be any subset of $P_0\setminus P_1$ of size $5k$. We know that for any $v\in P$ and $w\in Y$ we have that $d^+(w)>d^+(v)+4k$.

Consider semi-complete digraph $T[P]$. We have that the number of vertices with outdegrees at least $k$ in $T[P]$ is at least $k$, as otherwise the sum of outdegrees in $T[P]$ would be at most $k\cdot 5k+4k\cdot k=9k^2<\binom{5k}{2}$, so the sum of outdegrees would be strictly smaller than the number of arcs in the digraph. Let $Z$ be an arbitrary set of $k$ vertices with outdegrees at least $k$ in $T[P]$. We prove that $Z$ is a $(k,4)$-short immersion jungle.

Let us take any $v,w\in Z$; we are to construct $k$ edge-disjoint paths from $v$ to $w$ of length~$4$. Since the outdegree of $v$ in $T[P]$ is at least $k$, as the first vertices on the paths we can take any $k$ outneighbours of $v$ in $P$; denote them $v^1_1,v^1_2,\ldots,v^1_k$. By the definition of $P_0$, each $v^1_i$ is incident to some arc from $E(X,Y)$. As the second vertices on the paths we choose the heads of these arcs, denote them by $v^2_i$, thus constructing paths $v\to v^1_i\to v^2_i$ of length $2$ for $i=1,2,\ldots,k$. Note that all the arcs used for constructions so far are pairwise different.

We now consecutively finish paths $v\to v^1_i\to v^2_i$ using two more arcs in a greedy manner. Consider path $v\to v^1_i\to v^2_i$. As $v^2_i\in Y$ and $w\in P$, we have that $d^+(v^2_i)>d^+(w)+4k$. Hence, by Lemma~\ref{lem:degree-gap} we can identify $4k$ paths of length $2$ leading from $v^2_i$ to $w$. At most $2k$ of them contain an arc that was used in the first phase of the construction (two first arcs of the paths), and at most $2(i-1)\leq 2k-2$ of them can contain an arc used when finishing previous paths. This leaves us at least one path of length $2$ from $v^2_i$ to $w$ with no arc used so far, which we can use to finish the path $v\to v^1_i\to v^2_i$.
\end{proof}

\section{Algorithms for cutwidth}\label{sec:cutwidth}

In this section we present the algorithms for computing cutwidth. We start with the approximation algorithm and then proceed to the exact algorithm.

\begin{theorem}
Let $T$ be a semi-complete digraph. Then any outdegree ordering of $V(T)$ has width at most $m(\ctw(T))$, where $m(t)=100t^2+22t+1$.
\end{theorem}
\begin{proof}
Let $\sigma$ be any outdegree ordering of $V(T)$. If $\sigma$ had width more than $m(\ctw(T))$, then one of the partitions $(\sigma[\alpha],V(T)\setminus \sigma[\alpha])$ would be a $(m(\ctw(T))+1)$-backward tangle. Existence of such a structure is a contradiction with Lemma~\ref{lem:backward-tangle-cutwidth}.
\end{proof}

This gives raise to a straightforward approximation algorithm for cutwidth of a semi-complete digraph that simply sorts the vertices with respect to outdegrees, and then scans through the ordering checking whether it has small width. Note that this scan may be performed in $O(|V(T)|^2)$ time, as we maintain the cut between the prefix and the suffix of the ordering by iteratively moving one vertex from the suffix to the prefix.

\begin{theorem}
There exists an algorithm which, given a semi-complete digraph $T$ and an integer $k$, in time $O(|V(T)|^2)$ outputs an ordering of $V(T)$ of width at most $m(k)$ or a $(m(k)+1)$-backward tangle in $T$, where $m(t)=100t^2+22t+1$. In the second case the algorithm concludes that $\ctw(T)>k$.
\end{theorem}

We now present the exact algorithm for cutwidth.

\begin{theorem}\label{thm:cutwidth-exact}
There exists an algorithm, which given a semi-complete digraph $T$ and an integer $k$, in time $O(2^{O(k)}|V(T)|^2)$ outputs an ordering of $V(T)$ of width at most $k$, or correctly concludes that $\ctw(T)>k$.
\end{theorem}
\begin{proof}
Let $\cutfam$ be the family of such partitions $(X,Y)$ of $V(T)$ that $|E(X,Y)|\leq k$. Let us define an auxiliary digraph $D$ with $\cutfam$ as the vertex set, where $((X_1,Y_1),(X_2,Y_2))\in E(D)$ if and only if $X_2=X_1\cup \{v\}$ for some element $v\notin X_1$ and, consequently, $Y_1=Y_2\cup \{v\}$. Clearly, orderings of $V(T)$ of width at most $k$ correspond to paths from $(\emptyset, V(T))$ to $(V(T),\emptyset)$ in $D$. We show that if $\ctw(T)\leq k$, then $|D|=O(2^{O(k)}|V(T)|)$ and $D$ can be computed in $O(2^{O(k)}|V(T)|^2)$ time. Hence, a linear-time reachability algorithm applied in $D$ works within the claimed complexity bound.

Let $\sigma=(v_1,v_2,\ldots,v_n)$ be any outdegree ordering of $V(T)$, where $n=|V(T)|$. Such an ordering can be computed in $O(|V(T)|^2)$ time. The crucial observation is the following: if there is an index $i$ such that $d^+(v_{i+10k+1})\leq d^+(v_i)+2k$, then the set $\{v_i,v_{i+1},\ldots,v_{i+10k+1}\}$ is a $(10k+2,2k)$-degree tangle. By Lemma~\ref{lem:degree-tangle-pathwidth} we have that $\pw(T)>2k$, hence also $\ctw(T)>k$ and the algorithm may safely provide a negative answer. Note that checking whether such a situation occurs may be performed in $O(|V(T)|)$ time once the degrees and the degree ordering is computed, so from now on we assume that it does not occur: $d^+(v_{i+10k+1})>d^+(v_i)+2k$ for every index $i$.

The second observation is as follows: if $(X,Y)\in \cutfam$ and $x\in X, y\in Y$, then $d^+(x)\leq d^+(y)+k+1$. Otherwise, by Lemma~\ref{lem:degree-gap} there would be $k+1$ vertex-disjoint paths of length $2$ from $x$ to $y$; each of these paths would contribute with at least one arc to $E(X,Y)$, thus contradicting $|E(X,Y)|\leq k$.

From both observations we infer that if $(X,Y)\in \cutfam$ and $\alpha$ is the smallest index of an element of $Y$, then $\{v_1,\ldots,v_{\alpha-1}\}\subseteq X$, $\{v_{\alpha+10k+1},v_{\alpha+10k+2},\ldots,v_n\}\subseteq Y$, and the set $\{v_\alpha,\ldots,v_{\alpha+10k}\}$ is split between $X$ and $Y$. Now it becomes clear why $|D|\leq O(2^{O(k)}|V(T)|)$: each partition from $\cutfam$ may be characterized by $\alpha$ (equal to $|V(T)|+1$ if $Y=\emptyset$) and a bit mask of length $10k$ denoting, which elements of $\{v_{\alpha+1},\ldots,v_{\alpha+10k}\}$ belong to $X$ and which belong to $Y$. We also need to observe that the outdegrees in $D$ are bounded by $10k+1$: only vertices from the set $\{v_\alpha,\ldots,v_{\alpha+10k}\}$ can be moved from $Y$ to $X$ without violating the property from the second observation.

The digraph $D$ may be constructed in $O(2^{O(k)}|V(T)|^2)$ time as follows: 
\begin{itemize}
\item We generate $\cutfam$ by checking for each of $O(2^{O(k)}|V(T)|)$ characterizations by $\alpha$ and the bit mask if it corresponds to a partition from $\cutfam$. A naive way to do it requires an $O(|V(T)|^2)$ check for every characterization. However, we may reorganize the computation as follows. First, calculate the sizes of sets $E(\sigma[i],V(T)\setminus \sigma[i])$ by moving vertices from the right side to the left side in order of $\sigma$, beginning with partition $(\emptyset,V(T))$ and finishing with $(V(T),\emptyset)$. At each step we update $|E(\sigma[i],V(T)\setminus \sigma[i])|$ by considering only arcs incident to vertex $v_i$, which takes $O(|V(T)|)$ time. Then, for every $\alpha=1,2,\ldots,|V(T)|+1$ we check every bit mask by consecutively moving at most $10k$ vertices from the right side of the partition $(\sigma[\alpha-1],V(T)\setminus \sigma[\alpha-1])$ to the left side. Each such check takes $O(k|V(T)|)$ time as, again, during every move we consider only arcs incident to the moved vertex. We store $\cutfam$ in the $(\alpha,\textrm{bit mask})$ representation.
\item For every element of $\cutfam$ represented by $\alpha$ and a bit mask, we construct outgoing arcs by trying to move every single element of $\{v_\alpha,\ldots,v_{\alpha+10k}\}$ that is not in the left side of the partition (this can be retrieved from the bit mask) to the left side. Note that the characterization of the new partition by $\alpha'$ and a new bit mask can be computed in $O(k)$ time, and we can efficiently check whether this characterization corresponds to a partition that was verified to be in $\cutfam$.
\end{itemize}
Hence, the algorithm constructs the digraph $D$ and applies any linear-time reachability algorithm to test, whether $(V(T),\emptyset)$ can be reached from $(\emptyset,V(T))$. If this is the case, the path corresponds to an optimal ordering; otherwise we provide a negative answer.
\end{proof}

\section{Algorithms for pathwidth}\label{sec:pathwidth}

\subsection{Subset selectors for bipartite graphs}\label{sec:pathwidth-subset}

In this subsection we propose a formalism for expressing selection of a subset of vertices of a bipartite graph. We choose to introduce this formal layer, as in the approximation and exact algorithms for pathwidth we use two different such concepts that share some properties.

Let $\bipclass$ be the class of undirected bipartite graphs with fixed bipartition, expressed as triples: left side, right side, the edge set. Let $\mu(G)$ be the size of a maximum matching in $G$.

\begin{definition}
A function $f$ defined on $\bipclass$ is called a {\em{subset selector}} if $f(G)\subseteq V(G)$ for every $G\in\bipclass$. A {\em{reversed}} subset selector $f^{\rev}$ is defined as $f^{\rev}((X,Y,E))=f((Y,X,E))$. We say that subset selector $f$ is 
\begin{itemize}
\item a {\em{vertex cover selector}} if $f(G)$ is a vertex cover of $G$ for every $G\in \bipclass$, i.e., every edge of $G$ has at least one endpoint in $f(G)$;
\item {\em{symmetric}} if $f=f^{\rev}$;
\item {\em{monotonic}} if for every graph $G=(X,Y,E)$ and its subgraph $G'=G\setminus w$ where $w\in Y$, we have that $f(G)\cap X\supseteq f(G')\cap X$ and $f(G)\cap (Y\setminus \{w\}) \subseteq f(G')\cap Y$.
\end{itemize}
\end{definition}

The following observation expresses, how monotonic subset selectors behave with respect to modifications of the graph. By addition of a vertex we mean adding a new vertex to the vertex set, together with an arbitrary set of edges connecting it to the old ones. 

\begin{lemma}\label{lem:monotonicity}
Assume that $f$ and $f^{\rev}$ are monotonic subset selector and let $G=(X,Y,E)$ be a bipartite graph.
\begin{itemize}
\item If $v\in f(G)\cap Y$ then $v$ stays chosen by $f$ after any sequence of additions of vertices to the left side and deletions of vertices (different from $v$) from the right side.
\item If $v\in X\setminus f(G)$ then $v$ stays not chosen by $f$ after any sequence of additions of vertices to the left side and deletions of vertices from the right side.
\end{itemize}
\end{lemma}
\begin{proof}
For both claims, staying (not) chosen after a deletion on the right side follows directly from the definition of monotonicity of $f$. Staying (not) chosen after an addition on the left side follows from considering deletion of the newly introduced vertex and monotonicity of $f^{\rev}$.
\end{proof}

\subsubsection{The matching selector}

The subset selector that will be used for the approximation of pathwidth is the following:

\begin{definition}
By {\em{matching selector}} $\match$ we denote a subset selector that assigns to every bipartite graph $G$ the set of all the vertices of $G$ that are matched in {\bf{every}} maximum matching in $G$.
\end{definition}

Let us note that for any bipartite graph $G=(X,Y,E)$ we have that $|\match(G)\cap X|, |\match(G)\cap Y|\leq \mu(G)$. It appears that $\match$ is a symmetric and monotonic vertex cover selector. The symmetry is obvious. The crucial property of $\match$ is monotonicity: its proof requires technical and careful analysis of alternating and augmenting paths in bipartite graphs. $\match$ admits also an alternative characterization, expressed in Lemma~\ref{lem:characterization}: it can be computed directly from any maximum matching by considering alternating paths originating in unmatched vertices. This observation can be utilized to construct an algorithm that maintains $\match(G)$ efficiently during graph modifications. Moreover, from this alternative characterization it is clear that $\match$ is a vertex cover selector. The following lemma expresses all the vital properties of $\match$ that will be used in the approximation algorithm for pathwidth.

\begin{lemma}\label{lem:together}
$\match$ is a symmetric, monotonic vertex cover selector, which can be maintained together with a maximum matching of $G$ with updates times $O((\mu(G)+1)\cdot |V(G)|)$ during vertex additions and deletions. Moreover, $|\match(G)\cap X|, |\match(G)\cap Y|\leq \mu(G)$ for every bipartite graph $G=(X,Y,E)$.
\end{lemma}

We now proceed to the proof of Lemma~\ref{lem:together}; we assume reader's knowledge of basic concepts and definitions from the classical matching theory~(see \cite{diestel} for reference). We split the proof into several lemmas. First, we show that $\match$ is indeed monotonic.

\begin{lemma}\label{lem:match-mon}
$\match$ is monotonic.
\end{lemma}
\begin{proof}
Let $G'=G\setminus w$, where $G=(X,Y,E)$ is a bipartite graph and $w\in Y$. 

Firstly, we prove that $\match(G)\cap (Y\setminus \{w\}) \subseteq \match(G')\cap Y$. For the sake of contradiction, assume that there is some $v\in \match(G)\cap Y, v\neq w$, such that $v\notin \match(G')\cap Y$. So there exists a maximum matching $N$ of $G'$ in which $v$ is unmatched; we will also consider $N$ as a (possibly not maximum) matching in $G$. Let $M$ be any maximum matching in $G$; as $v\in \match(G)$ we have that $v$ is matched in $M$. Construct a maximum path $P$ in $G$ that begins in $v$ and alternates between matchings $M$ and $N$, starting with $M$. As both $w$ and $v$ are not matched in $N$, this path does not enter $w$ or $v$, hence it ends somewhere else; note that $P$ has length at least one as $v$ is matched in $M$. Let $x\notin \{v,w\}$ be the last vertex of $P$. We consider two cases.

Assume first that $x\in X$, i.e., $x$ is a vertex of the left side that is not matched in $N$. Then $P$ is an augmenting path for $N$ fully contained in $G'$. This contradicts maximality of $N$.

Assume now that $x\in Y$, i.e., $x$ is a vertex of the right side that is not matched in $M$. Then $P$ is an alternating path for $M$ in $G$, whose switching leaves $v$ unmatched. This contradicts the fact that $v\in \match(G)$.

Now we prove that $\match(G)\cap X\supseteq \match(G')\cap X$. We consider two cases:

Assume first that $w\in \match(G)$. As $w$ is matched in every maximum matching of $G$, after deleting $w$ the size of the maximum matching drops by one: if there was a matching of the same size in $G'$, it would constitute also a maximum matching in $G$ that does not match $w$. Hence, if we take any matching $M$ of $G$ and delete the edge incident to $w$, we obtain a maximum matching of $G'$. We infer that $\match(G)\cap X\supseteq \match(G')\cap X$, as every vertex of $X$ that is unmatched in some maximum matching in $G$ is also unmatched in some maximum matching $G'$. 

Assume now that $w\notin \match(G)$. Let $M$ be any maximum matching of $G$ in which $w$ is unmatched. As $M$ is also a matching in $G'$, we infer that $M$ is also a maximum matching in $G'$ and the sizes of maximum matchings in $G$ and $G'$ are equal. Take any $v\in \match(G')\cap X$ and for the sake of contradiction assume that $v\notin \match(G)$. Let $N$ be any maximum matching in $G$ in which $v$ is unmatched. Note that since $v\in \match(G')\cap X$ and $M$ is a maximum matching in $G'$, then $v$ is matched in $M$. Let us now construct a maximum path $P$ in $G$ that begins in $v$ and alternates between $M$ and $N$, starting with $M$. As $v$ is unmatched in $N$ and $w$ is unmatched in $M$, this path does not enter $w$ or $v$, hence it ends somewhere else; note that $P$ has length at least one as $v$ is matched in $M$. Let $x\notin \{v,w\}$ be the last vertex of $P$. Again, we consider two subcases.

In the first subcase we have $x\in X$, i.e., $x$ is a vertex of the left side that is not matched in $M$. Then $P$ is an alternating path for $M$ in $G'$, whose switching leaves $v$ unmatched. This contradicts $v\in \match(G')$.

In the second subcase we have $x\in Y$, i.e., $x$ is a vertex of the right side that is not matched in $N$. Then $P$ is an augmenting path for $N$ in $G$, which contradicts maximality of $N$.
\end{proof}

In order to prove that $\match$ is a vertex cover selector and can be computed efficiently, we prove the following alternative characterization. The following lemma might be considered a folklore corollary of the classical matching theory, but for the sake of completeness we include its proof. By $V(F)$ for $F\subseteq E(G)$ we denote the set of endpoints of edges of $F$.

\begin{lemma}\label{lem:characterization}
Let $G=(X,Y,E)$ be a bipartite graph and let $M$ be any maximum matching in $G$. Let $A_0=X\setminus V(M)$ be the set of unmatched vertices on the left side, and $B_0=Y\setminus V(M)$ be the set of unmatched vertices on the right side. Moreover, let $A$ be the set of vertices of $X$ that can be reached via an alternating path (with respect to $M$) from $A_0$, and symmetrically let $B$ be the set of vertices of $Y$ reachable by an alternating path from $B_0$. Then $\match(G)=V(G)\setminus(A\cup B)$.
\end{lemma}
\begin{proof}
On one hand, every vertex $v$ belonging to $A$ or $B$ is not matched in some maximum matching: we just modify $M$ along the alternating path connecting $v$ with a vertex unmatched in $M$, obtaining another maximum matching in which $v$ is unmatched. Hence, $\match(G)\subseteq V(G)\setminus(A\cup B)$.

We now proceed to the proof that $\match(G)\supseteq V(G)\setminus(A\cup B)$. Assume that there is an edge between $A$ and $B$, i.e., there is a vertex $w$ in the set $N(A)\cap B$. Let $v$ be the neighbour of $w$ in $A$. We claim that $w$ is reachable from $A_0$ via an alternating path. Assume otherwise, and take an alternating path $P$ from $A_0$ reaching $v$. Observe that this path does not traverse $w$ and, moreover, $vw\notin M$, as then this edge would be used in $P$ to access $v$. Hence, we can prolong $P$ by the edge $vw$, thus reaching $w$ by an alternation path from $A_0$, a contradiction. By the definition of $B$, $w$ is also reachable also from $B_0$ via an alternating path. The concatenation of these two paths is an alternating walk from $A_0$ to $B_0$, which contains an alternating simple subpath from $A_0$ to $B_0$. This subpath is an augmenting path for $M$, which contradicts maximality of $M$. 

Consider the set $A$ and the set $N(A)$. We already know that $N(A)$ is disjoint with $B$, so also from $B_0$; hence, every vertex of $N(A)$ must be matched in $M$. Moreover, we have that every vertex of $N(A)$ must be in fact matched to a vertex of $A$, as the vertices matched to $N(A)$ are also reachable via alternating paths from $A_0$. Similarly, every vertex of $N(B)$ is matched to a vertex of $B$. 

We now claim that $|X\setminus A|+|N(A)|=|M|$. As $A_0\subseteq A$, every vertex of $X\setminus A$ as well as every vertex of $N(A)$ is matched in $M$. Moreover, as there is no edge between $A$ and $Y\setminus N(A)$, every edge of $M$ has an endpoint either in $X\setminus A$ or in $N(A)$. It remains to show that no edge of $M$ can have one endpoint in $X\setminus A$ and second in $N(A)$; this, however follows from the fact that vertices of $N(A)$ are matched to vertices of $A$, proved in the previous paragraph. Similarly we have that $|Y\setminus B|+|N(B)|=|M|$.

It follows that sets $(X\setminus A)\cup N(A)$ and $N(B)\cup (X\setminus B)$ are vertex covers of $G$ of size $|M|$. Note that every vertex cover $C$ of $G$ of size $|M|$ must be fully matched by any matching $N$ of size $|M|$, as every edge of $N$ is incident to at least one vertex from $C$. Hence, every vertex of $(X\setminus A)\cup N(A) \cup N(B)\cup (X\setminus B)=V(G)\setminus(A\cup B)$ is matched in every maximum matching of $G$, so $\match(G)\supseteq V(G)\setminus(A\cup B)$.
\end{proof}

From now on we use the terminology introduced in Lemma~\ref{lem:characterization}. Note that Lemma~\ref{lem:characterization} implies that sets $A,B$ do not depend on the choice of matching $M$, but only on graph $G$. Also, from the proof of Lemma~\ref{lem:characterization} it follows that $\match$ is a vertex cover selector: there is no edge between $A$ and $B$, because it would close an augmenting path from $A_0$ to $B_0$. We now present an incremental algorithm that maintains $\match(G)$ together with a maximum matching of $G$ during vertex additions and deletions.

\begin{lemma}\label{lem:match-maintaining}
There exists an incremental algorithm, which maintains a maximum matching $M$ and $\match(G)$ for a bipartite graph $G=(X,Y,E)$ during vertex addition and deletion operations. The update time is $O((|M|+1)\cdot |V(G)|)$.
\end{lemma}
\begin{proof}
Observe that, by Lemma~\ref{lem:characterization}, $\match(G)$ can be computed from $M$ in $O(|G|)$ time: we simply compute sets $A_0,B_0$ and apply breadth-first search from the whole $A_0$ to compute $A$, and breadth-first search from the whole $B_0$ to compute $B$. Both of these searches take $O(|G|)$ time; note that by K\"onig's theorem a bipartite graph with maximum matching $M$ can have at most $O(|M|\cdot |V(G)|)$ edges, so $|G|=O((|M|+1)\cdot |V(G)|)$. Hence, we just need to maintain a maximum matching.

During vertex addition, the size of the maximum matching can increase by at most $1$, so we may simply add the new vertex and run one iteration of the standard breadth-first search procedure checking, whether there is an augmenting path in the new graph; this takes $O(|G|)$ time. If this is the case, we modify the matching along this path and we know that we obtained a maximum matching in the new graph. Otherwise, the size of the maximum matching does not increase, so we do not need to modify the current one. Similarly, during vertex deletion we simply delete the vertex together with possibly at most one edge of the matching incident to it. The size of the stored matching might have decreased by~$1$; in this case again we run one iteration of checking whether there is an augmenting path, again in $O(|G|)$ time. If this is the case, we augment the matching using this path, obtaining a new matching about which we know that it is maximum. Otherwise, we know that the current matching is maximum, so we do not need to modify it.\end{proof}

Lemmas~\ref{lem:match-mon},~\ref{lem:characterization} and~\ref{lem:match-maintaining} together with previous observations prove Lemma~\ref{lem:together}.

\subsubsection{The Buss selector}

In this subsection we introduce the subset selector that will be used in the exact algorithm for pathwidth. This selector is inspired by the classical kernelization algorithm for the vertex cover problem of Buss~\cite{buss}. By $d(v)$ we denote the (undirected) degree of vertex $v$.

\begin{definition}
Let $G=(X,Y,E)$ be a bipartite graph. A vertex $v$ is called $\ell${\em{-important}} if $d(v)>\ell$, and $\ell${\em{-unimportant}} otherwise. A {\em{Buss selector}} is a subset selector $\buss_\ell$ that returns all vertices of $X$ that are either $\ell$-important, or have at least one $\ell$-unimportant neighbour.
\end{definition}

Note that Buss selector is highly non-symmetric, as it chooses vertices only from the left side. However, both $\buss_\ell$ and $\buss_\ell^\rev$ behave in a nice manner.

\begin{lemma}\label{lem:buss-mon}
Both $\buss_\ell$ and $\buss_\ell^{\rev}$ are monotonic.
\end{lemma}
\begin{proof}
Monotonicity of $\buss_\ell$ is equivalent to the observation that if $v\in X$ is $\ell$-unimportant and has only $\ell$-important neighbours, then deletion of any vertex of $Y$ cannot make $v$ $\ell$-important or create $\ell$-unimportant neighbours of $v$. This holds because the degrees of surviving neighbours of $v$ do not change.

Monotonicity of $\buss_\ell^\rev$ is equivalent to the observation that if $w\in Y$ is chosen by $\buss_\ell^\rev$ because it is $\ell$-important, then after deletion of any other $w'\in Y$ its degree does not change so it stays $\ell$-important, and if it had an $\ell$-unimportant neighbour, then after deletion of any other $w'\in Y$ this neighbour will still be $\ell$-unimportant.
\end{proof}

We now prove that $\buss_\ell$ does not choose too many vertices unless $G$ contains a large matching.

\begin{lemma}\label{lem:buss-retrieving}
If $|\buss_\ell(G)|>\ell^2+\ell$, then $G$ contains a matching of size $\ell+1$.
\end{lemma}
\begin{proof}
As $|\buss_\ell(G)|>\ell^2+\ell$, in $\buss_\ell(G)$ there are at least $\ell+1$ $\ell$-important vertices, or at least $\ell^2+1$ vertices with an $\ell$-unimportant neighbour. In both cases we construct the matching greedily.

In the first case we iteratively take an $\ell$-important vertex of $\buss_\ell(G)$ and match it with any its neighbour that is not matched so far. As there are at least $\ell+1$ these neighbours and at most $\ell$ were used so far, we can always find one not matched so far.

In the second case we take any vertex $v_1$ of $\buss_\ell(G)$ and find any its $\ell$-unimportant neighbour $w_1$. We add $v_1w_1$ to the constructed matching and mark all the at most $\ell$ neighbours of $w_1$ as used. Then we take any unused vertex $v_2$ of $\buss_\ell(G)$ that has an $\ell$-unimportant neighbour, find any its $\ell$-unimportant neighbour $w_2$ (note that $w_2\neq w_1$ as $v_2$ was not marked), add $v_2w_2$ to the constructed matching and mark all the at most $\ell$ neighbours of $w_2$ as used. We continue in this manner up to the point when a matching of size $\ell+1$ is constructed. Note that there will always be an unmarked vertex of $\buss_\ell(G)$ with an $\ell$-unimportant neighbour, as at the beginning there are at least $\ell^2+1$ of them and after $i$ iterations at most $i\cdot\ell$ are marked as used.
\end{proof}

We prove that $\buss_\ell$ can be also evaluated efficiently.

\begin{lemma}\label{lem:buss-maintaining}
There exists an incremental algorithm, which maintains $\buss_\ell(G)$ for a bipartite graph $G=(X,Y,E)$ during operations of vertex addition and vertex deletion. The update times are $O(\ell|V(G)|)$.
\end{lemma}
\begin{proof}
With every vertex of $X$ we maintain its degree and a counter of $\ell$-unimportant neighbours. We also maintain degrees of vertices of $Y$. The degree gives us information whether the vertex is $\ell$-important. 

Let us first examine adding vertex $v$ to the left side. We need to increase the degrees of neighbours of $v$, so some of them may become $\ell$-important. For every such vertex that became $\ell$-important --- note that its degree is exactly $\ell+1$ --- we examine its $\ell+1$ neighbours and decrement their $\ell$-unimportant neighbours' counters. If it drops to zero, we delete this vertex from $\buss_\ell(G)$ unless it's $\ell$-important. Finally, we count how many neighbours of $v$ are $\ell$-unimportant, set the counter appropriately and add $v$ to $\buss_\ell(G)$ if necessary.

Let us now examine adding vertex $w$ to the right side. We need to increase the degrees of neighbours of $w$, so some of them may become $\ell$-important and thus chosen by $\buss_\ell(G)$. Moreover, if $w$ is $\ell$-unimportant, then we increment the $\ell$-unimportant neighbours' counters of neighbours of $w$; if it is incremented from $0$ to $1$, then the vertex becomes chosen by $\buss_\ell(G)$, assuming that it was not already $\ell$-important.

Now we examine deleting vertex $v$ from the left side. We iterate through the neighbours of $v$ and decrement their degrees. If any of them ceases to be $\ell$-important, we iterate through all its $\ell$ neighbours and increment the counters of $\ell$-unimportant neighbours. If for some of these neighbours the counter was incremented from $0$ to $1$, the neighbour becomes chosen by $\buss_\ell(G)$, assuming that it was not already $\ell$-important.

Finally, we examine deleting vertex $w$ from the right side. Firstly, we check if $w$ was $\ell$-important. Then we iterate through neighbours of $w$ decreasing the degree and $\ell$-unimportant neighbours' counters, if necessary. If any of the neighbours ceases to be $\ell$-important or have an $\ell$-unimportant neighbour, it becomes not chosen to $\buss_\ell(G)$.
\end{proof}

\subsection{Algorithms}\label{sec:pathwidth-algo}

In this subsection we present the algorithms for computing pathwidth. We begin with approximation algorithm and then proceed to the exact algorithm. We introduce the approximation algorithm with an additional parameter $\ell$; taking $\ell=5k$ gives the promised $7$-approximation, but modifying the parameter $\ell$ may be useful to improve quality of the obtained degree tangle.

\begin{theorem}\label{thm:pathwidth-approx}
There exists an algorithm, which given a semi-complete digraph $T$ and integers $k$ and $\ell\geq 5k$, in time $O(k|V(T)|^2)$ outputs one of the following:
\begin{itemize}
\item an $(\ell+2,k)$-degree tangle in $T$;
\item a $(k+1,k)$-matching tangle in $T$;
\item a path decomposition of $T$ of width at most $\ell+2k$.
\end{itemize}
In the first two cases the algorithm can correctly conclude that $\pw(T)>k$.
\end{theorem}
\begin{proof}
The last sentence follows from Lemmas~\ref{lem:degree-tangle-pathwidth} and~\ref{lem:matching-tangle-pathwidth}. We proceed to the algorithm.

The algorithm first computes any outdegree ordering $\sigma=(v_1,v_2,\ldots,v_n)$ of $V(T)$ in $O(|V(T)|^2)$ time, where $n=|V(T)|$. Then in $O(|V(T)|)$ time we check if there is an index $i$ such that $d^+(v_{i+\ell+1})\leq d^+(v_i)+k$. If this is true, then $\{v_i,v_{i+1},\ldots,v_{i+\ell+1}\}$ is an $(\ell+2,k)$-degree tangle which can be safely output by the algorithm. From now on we assume that such a situation does not occur, i.e, $d^+(v_{i+\ell+1})>d^+(v_i)+k$ for every index $i$.

We define a separation sequence $R_0=((A_0,B_0),(A_1,B_1),\ldots,(A_{n-\ell},B_{n-\ell}))$ as follows. Let us define $S^0_i=\{v_{i+1},v_{i+2},\ldots,v_{i+\ell}\}$ and let $H_i=(X_i,Y_i,E_i)$ be a bipartite graph, where $X_i=\{v_1,\ldots,v_i\}$, $Y_i=\{v_{i+\ell+1},v_{i+\ell+2},\ldots,v_n\}$ and $xy\in E_i$ if and only if $(x,y)\in E(T)$. If $\mu(H_i)>k$, then vertices matched in a maximum matching of $H_i$ form a $(k+1,k)$-matching tangle in $T$, which can be safely output by the algorithm. Otherwise, let $S_i=S^0_i\cup \match(H_i)$ and we set $A_i=X_i\cup S_i$ and $B_i=Y_i\cup S_i$; the fact that $(A_i,B_i)$ is a separation follows from the fact that $\match$ is a vertex cover selector. Finally, we add separations $(\emptyset,V(T))$ and $(V(T),\emptyset)$ at the ends of the sequence, thus obtaining separation sequence $R$. We claim that $R$ is a separation chain. Note that if we prove it, the width of the corresponding path decomposition is upper bounded by $\max_{0\leq i\leq n-\ell-1}|\{v_{i+1},v_{i+2},\ldots,v_{i+\ell+1}\}\cup (\match(H_i)\cap X_i) \cup (\match(H_{i+1})\cap Y_{i+1})|-1\leq \ell+1+2k-1=\ell+2k$, by monotonicity of $\match$.

It suffices to show that for every $i$ we have that $A_i\subseteq A_{i+1}$ and $B_i\supseteq B_{i+1}$. This, however, follows from Lemma~\ref{lem:monotonicity} and the fact that $\match$ is monotonic. $H_{i+1}$ differs from $H_i$ by deletion of one vertex on the right side and addition of one vertex on the left side, so we have that $A_{i+1}$ differs from $A_i$ only by possibly incorporating vertex $v_{i+\ell+1}$ and some vertices from $Y_{i+1}$ that became chosen by $\match$, and $B_{i+1}$ differs from $B_i$ only by possibly losing vertex $v_{i+1}$ and some vertices from $X_i$ that ceased to be chosen by $\match$.   

Separation chain $R$ can be computed in $O(k|V(T)|^2)$ time: we consider consecutive sets $S^0_i$ and maintain the graph $H_i$ together with a maximum matching in it and $\match(H_i)$. As going to the next set $S^0_i$ can be modelled by one vertex deletion and one vertex additions in graph $H_i$, by Lemma~\ref{lem:together} we have that the time needed for an update is $O(k|V(T)|)$; note that whenever the size of the maximum matching exceeds $k$, we terminate the algorithm by outputting the obtained matching tangle. As we make $O(|V(T)|)$ updates, the time bound follows. Translating a separation chain into a path decomposition can be done in $O(\ell|V(T)|)$ time, assuming that we store the separators along with the separations.
\end{proof}

We now present the exact algorithm for pathwidth.

\begin{theorem}\label{thm:pathwidth-exact}
There exists an algorithm that, given a semi-complete digraph $T$ and an integer $k$, in $O(2^{O(k\log k)}|V(T)|^2)$ time computes a path decomposition of $T$ of width at most $k$, or correctly concludes that no such exists.
\end{theorem}
\begin{proof}
We say that a separation chain $R$ is {\em{feasible}} if it corresponds to a path decomposition of width at most $k$ in the sense of the second claim of Lemma~\ref{lem:sepchain}. Note that the transformations in the first and in the second claim of Lemma~\ref{lem:sepchain} are inverse to each other, so by Lemma~\ref{lem:sepchain} we may look for a feasible separation chain. Turning such a separation chain into a path decomposition can be trivially done in $O(k|V(T)|)$ time, assuming that we store the separators along with the separations.

We define a subclass of separation chains called {\em{thin}}. Every feasible separation chain can be adjusted to a thin separation chain of at most the same width by deleting some vertices from sets $A_i,B_i$; note that the new separation chain created in such a manner will also be feasible, as bags of the corresponding path decomposition can only get smaller. Hence, we may safely look for a separation chain that is thin.

The opening step of the algorithm is fixing some outdegree ordering $\sigma=\{v_1,v_2,\ldots,v_n\}$ of $V(T)$, where $n=|V(T)|$. For $i=0,1,\ldots,n$, let $H_i$ be a bipartite graph with left side $X_i=\{v_1,v_2,\ldots,v_i\}$ and right $Y_i=\{v_{i+1},v_{i+2},\ldots,v_n\}$, where $xy\in E(H_i)$ if $(x,y)\in E(T)$. As in the proof of Theorem~\ref{thm:pathwidth-approx}, in $O(|V(T)|)$ time we check if there is an index $i$ such that $d^+(v_{i+5k+1})\leq d^+(v_i)+k$. If yes, then $\{v_i,v_{i+1},\ldots,v_{i+5k+1}\}$ is a $(5k+2,k)$-degree tangle and the algorithm may safely provide a negative answer. From now on we assume that such a situation does not occur.

We proceed to the definition of thinness. We fix $m=6k+1$. Let $(A,B)$ be any separation in $T$ and $\alpha,\beta$ be any indices between $0$ and $|V(T)|$ such that $\alpha=\max(0,\beta-(5k+1))$.  We say that $(A,B)$ is {\em{thin with respect to}} $(\alpha,\beta)$  if {\em{(i)}} $Y_\beta\subseteq B\subseteq Y_\alpha\cup \buss_m(H_{\alpha})$ and {\em{(ii)}} $X_\alpha \subseteq A\subseteq X_\beta \cup \buss^\rev_m(H_{\beta})$. $(A,B)$ is {\em{thin}} if it is thin with respect to any such pair $(\alpha,\beta)$. A separation chain is {\em{thin}} if every its separation is thin.

For a separation $(A,B)$ let us define the {\em{canonical index}} $\beta=\beta((A,B))$ as the only integer between $0$ and $|V(T)|$ such that $d^+(v_j)<|A|$ for $j\leq \beta$ and $d^+(v_j)\geq |A|$ for $j>\beta$. Similarly, the {\em{canonical index}} $\alpha=\alpha((A,B))$ is defined as $\alpha((A,B))=\max(0,\beta((A,B))-(5k+1))$. Observe that by the assumed properties of ordering $\sigma$ we have that vertices in $X_\alpha$ have outdegrees smaller than $|A|-k$.

Firstly, we observe that if $(A,B)$ is a separation and $\alpha,\beta$ are its canonical indices, then $X_\alpha\subseteq A$ and $Y_\beta\subseteq B$. This follows from the fact that vertices in $A\setminus B$ have outdegrees smaller than $|A|$, hence they cannot be contained in $Y_\beta$, while vertices in $B\setminus A$ have outdegrees at least $|A\setminus B|\geq |A|-k$, hence they cannot be contained in $X_\alpha$. Concluding, vertices of $X_\alpha$ may belong only to $A\setminus B$ or $A\cap B$ (left side or the separator), vertices of $Y_\beta$ may belong only to $B\setminus A$ or $A\cap B$ (right side or the separator), while for vertices of the remaining part $Y_\alpha\cap X_\beta$ neither of the three possibilities is excluded.

We now show how to transform any separation chain into a thin one. Let us take some separation chain $R$ and obtain a sequence of separations $R'$ as follows. We take every separation $(A,B)$ from $R$; let $\alpha,\beta$ be its canonical indices. We delete $X_{\alpha}\setminus \buss_m(H_{\alpha})$ from $B$ and $Y_\beta\setminus \buss^\rev_m(H_{\beta})$ from $A$, thus obtaining a new pair $(A',B')$. We need to prove that every such pair $(A',B')$ is a separation, and that all these separations form a separation chain. The fact that such a separation chain is thin follows directly from the definition of the performed operation and the observation of the previous paragraph.

First, we check that $A'\cup B'=V(T)$. This follows from the fact from $A$ we remove only vertices of $Y_\beta$ while from $B$ we remove only vertices from $X_\alpha$, but after the removal $X_\alpha$ is still covered by $A'$ and $Y_\beta$ by $B'$.

Now we check that $E(A'\setminus B',B'\setminus A')=\emptyset$. Assume otherwise, that there is a pair $(v,w)\in E(T)$ such that $v\in A'\setminus B'$ and $w\in B'\setminus A'$. By the construction of $(A',B')$ and the fact that $(A,B)$ was a separation we infer that either $v\in X_{\alpha}\setminus\buss_m(H_{\alpha})$ or $w\in Y_{\beta}\setminus \buss^\rev_m(H_{\beta})$. We consider the first case, as the second is symmetrical.

Since $w\notin A'$ and $X_\alpha\subseteq A'$, we have that $w\in Y_\alpha$, so $vw$ is an edge in $H_\alpha$. As $v\notin \buss_m(H_\alpha)$, we have that $v$ is $m$-unimportant and has only $m$-important neighbours. Hence $w$ is $m$-important in $H_\alpha$. Observe now that $w$ cannot be contained in $B\setminus A$, as there is more than $m>k$ vertices in $X_\alpha$ being tails of arcs directed toward $w$, and only $k$ of them can be in separator $A\cap B$ leaving at least one belonging to $A\setminus B$ (recall that vertices from $X_\alpha$ cannot belong to $B\setminus A$). Hence $w\in A$. As $w\notin A'$, we have that $w\in Y_{\beta}\setminus \buss^\rev_m(H_{\beta})$. However, $w$ was an $m$-important vertex on the right side of $H_\alpha$, so as it is also on the right side of $H_\beta$, it is also $m$-important in $H_\beta$. This is a contradiction with $w\notin \buss^\rev_m(H_{\beta})$.

We conclude that $(A',B')$ is indeed a separation.

Finally, we check that $R'$ is a separation chain. Consider two separations $(A_1,B_1)$, $(A_2,B_2)$ in $R$, such that $A_1\subseteq A_2$ and $B_1\supseteq B_2$. Let $\alpha_1,\beta_1,\alpha_2,\beta_2$ be canonical indices of  $(A_1,B_1)$ and $(A_2,B_2)$, respectively. It follows that $\alpha_1\leq \alpha_2$ and $\beta_1\leq \beta_2$. Hence, graph $H_{\alpha_2}$ can be obtained from $H_{\alpha_1}$ via a sequence of vertex deletions on the right side and vertex additions on the left side. As $\buss_m$ and $\buss_m^\rev$ are monotonic (Lemma~\ref{lem:buss-mon}), by Lemma~\ref{lem:monotonicity} we have that every vertex deleted from $B_1$ while constructing $B_1'$ is also deleted from $B_2$ while constructing $B_2'$ (assuming it belongs to $B_2$). Hence, $B_2'\subseteq B_1'$. A symmetric argument shows that $A_2'\supseteq A_1'$.

We proceed to the algorithm itself. As we argued, we may look for a thin feasible separation chain.

Let $\cutfam$ be the family of thin separations of $T$ of order at most $k$. We construct an auxiliary digraph $D$ with vertex set $\cutfam$ by putting an arc $((A,B),(A',B'))\in E(D)$ if and only if $A\subseteq A'$, $B\supseteq B'$ and $|A'\cap B|\leq k+1$. By Lemma~\ref{lem:sepchain}, paths in $D$ from $(\emptyset, V(T))$ to $(V(T),\emptyset)$ correspond to feasible thin separation chains.

We prove that either the algorithm can find an obstacle for admitting path decomposition of width at most $k$, or $D$ has size at most $O(2^{O(k\log k)}|V(T)|)$ and can be constructed in time $O(2^{O(k\log k)}|V(T)|^2)$. Hence, any linear-time reachability algorithm in $D$ runs within claimed time complexity bound.

Consider any indices $\alpha,\beta$ such that $\alpha=\max(0,\beta-(5k+1))$. Observe that if $|\buss_m(H_\alpha)|>m^2+m$, by Lemma~\ref{lem:buss-retrieving} we can find a matching of size $m+1$ in $H_\alpha$. At most $5k+1=m-k$ edges of this matching have the right endpoint in $Y_\alpha\cap X_\beta$, which leaves us at least $k+1$ edges between $X_\alpha$ and $Y_\beta$. Such a structure is a $(k+1,k)$-matching tangle in $T$, so by Lemma~\ref{lem:matching-tangle-pathwidth} the algorithm may provide a negative answer. A symmetrical reasoning shows that if $|\buss^\rev_m(H_\beta)|>m^2+m$, then the algorithm can also provide a negative answer.

If we assume that these situations do not occur, we can characterize every thin separation $(A,B)$ of order at most $k$ by:
\begin{itemize}
\item a number $\beta$, where $0\leq \beta\leq |V(T)|$;
\item a mask on the vertices from $Y_\alpha\cap X_\beta$, denoting for each of them whether it belongs to  $A\setminus B$, $A\cap B$ or to $B\setminus A$ (at most $3^{5k+1}$ options);
\item subsets of size at most $k$ of $X_\alpha\cap \buss_m(H_\alpha)$ and $Y_\beta\cap \buss^\rev_m(H_\beta)$, denoting which vertices belong to $A\cap B$ (at most $\left(k\binom{O(k^2)}{k}\right)^2=2^{O(k\log k)}$ options).
\end{itemize}
Moreover, if $(A',B')$ is an outneighbour of $(A,B)$ in $D$, then it must have parameter $\beta'$ not larger than $\beta+(6k+2)$, as otherwise we have a guarantee that $|A'\cap B|\geq |X_{\alpha'}\cap Y_\beta|\geq k+2$, and also not smaller than $\beta-(6k+2)$, as otherwise we have a guarantee that $|A|\geq |X_\alpha|>|A'|$. Hence, the outdegrees in $D$ are bounded by $2^{O(k\log k)}$.

This gives raise to the following algorithm constructing $D$ in time $O(2^{O(k\log k)}|V(T)|^2)$.
\begin{itemize}
\item First, we enumerate $\cutfam$. We scan through the order $\sigma$ with an index $\beta$ maintaining graphs $H_\alpha, H_\beta$ for $\alpha=\max(0,\beta-(5k+1))$, along with $\buss_m(H_\alpha)$ and $\buss^\rev_m(H_\beta)$. Whenever cardinality of any of these sets exceeds $m^2+m$, we terminate the algorithm providing a negative answer. By Lemma~\ref{lem:buss-maintaining} we can bound the update time by $O(k|V(T)|)$. For given index $\beta$, we list all $2^{O(k\log k)}$ pairs $(A,B)$ having this particular $\beta$ in characterization from the previous paragraph. For every such pair, in $O(k|V(T)|)$ we check whether it induces a separation of order $k$, by testing emptiness of set $E(A\setminus B,B\setminus A)$ using at most $O(k)$ operations of vertex deletion/addition on the graph $H_\alpha$. We discard all the pairs that do not form such a separation; all the remaining ones are exactly separations that are thin with respect to $(\alpha,\beta)$.
\item For every separation $(A,B)$ characterized by parameter $\beta$, we check for all the $2^{O(k\log k)}$ separations $(A',B')$ with parameter $\beta'$ between $\beta-(6k+2)$ and $\beta+(6k+2)$, whether we should put an arc from $(A,B)$ to $(A',B')$. Each such a check can be performed in $O(k)$ time, assuming that we store the separator along with the separation.
\end{itemize}
Having constructed $D$, we run a linear-time reachability algorithm to check whether $(V(T),\emptyset)$ can be reached from $(\emptyset,V(T))$. If not, we provide a negative answer; otherwise, the path corresponds to a feasible separation chain which can be transformed into a path decomposition.
\end{proof}

\section{Topological containment and immersion}\label{sec:top}

\begin{theorem}
There exists an algorithm which, given a semi-complete $T$ and a digraph $H$ with $k=|H|$, in time $O(2^{O(k\log k)}|V(T)|^2)$ checks whether $H$ is topologically contained in $T$.
\end{theorem}
\begin{proof}
We run the algorithm given by Theorem~\ref{thm:pathwidth-approx} for parameters $20k$ and $520k$, which either returns a $(520k+2,20k)$-degree tangle, a $(20k+1,20k)$-matching tangle or a decomposition of width at most $560k$. If the last is true, we run the dynamic programming routine, which works in $O(2^{O(k\log k)}|V(T)|)$ time. The routine can be constructed in a similar manner to the one for immersion in~\cite{my}; for the sake of completeness, we include a full description in Appendix~\ref{app:dp}. However, if the approximation algorithm returned an obstacle, by Lemmas~\ref{lem:degree-tangle-jungle},~\ref{lem:matching-tangle-jungle} and~\ref{lem:sj-ts} we can provide a positive answer: existence of a $(520k+2,20k)$-degree tangle or a $(20k+1,20k)$-matching tangle ensures that $H$ is topologically contained in~$T$.
\end{proof}

By plugging in the dynamic programming routine for immersion of~\cite{my} instead of topological containment, we obtain the following.

\begin{theorem}
There exists an algorithm which given a semi-complete $T$ and a digraph $H$ with $k=|H|$, in time $O(2^{O(k^2\log k)}|V(T)|^2)$ checks whether $H$ can be immersed into $T$.
\end{theorem}

Finally, we can reduce the polynomial factor in the running time of the algorithm for rooted immersion of~\cite{my} by substituting the approximation routine for pathwidth with the one given by Theorem~\ref{thm:pathwidth-approx}. We remark that one needs to be a bit careful with running times of all the components of the algorithm; therefore, we include a full proof.

\begin{theorem}\label{thm:rimmersion}
There exists an algorithm which given a rooted semi-complete $\mathbf{T}$ and a rooted digraph $\mathbf{H}$ with $k=|H|$, in time $O(f(k)|V(T)|^3)$ checks whether $\mathbf{H}$ is a rooted immersion of in $\mathbf{T}$, for some elementary function $f$.
\end{theorem}
\begin{proof}
Let $f$ be the function given by Lemma~2 of~\cite{my}; i.e., basing on a $f(t)$-jungle in any semi-complete digraph $S$, one can find a $t$-triple in $S$ in time $O(|V(S)|^3\log |V(S)|)$. Moreover, let $p$ be the polynomial given by Lemma~$6$ of~\cite{my}; i.e., in a $p(|H|)$-triple that is disjoint from the roots one can find an irrelevant vertex for the {\sc{Rooted Immersion}} problem in $O(p(|H|)^2|V(T)|^2)$ time.

Given the input semi-complete rooted digraph $T$ and a rooted digraph $H$, we run the approximation algorithm of Theorem~\ref{thm:pathwidth-approx} for parameters $20k$ and $520k$ on $T$ with the roots removed, where $k=f(p(|H|))$; this takes $O(g(|H|)\cdot |V(T)|^2)$ time for some elementary function $g$. If the algorithm returns a decomposition of $T$ without roots of width at most $560k$, we include all the roots in every bag of the decomposition and finalize the algorithm by running the dynamic programming routine for rooted immersion from~\cite{my}, which takes $O(h(|H|)\cdot |V(T)|)$ time for some elementary function $h$. Otherwise, using Lemma~\ref{lem:degree-tangle-jungle} or Lemma~\ref{lem:matching-tangle-jungle} we extract a $(k,4)$-short jungle $X$ from the output $(520k+2,20k)$-degree tangle or $(20k+1,20k)$-matching tangle; this takes $O(k^{3}|V(T)|^2)$ time. 

Obviously, $X$ is also a $k$-jungle in the sense of Fradkin and Seymour, so we are tempted to run the algorithm of Lemma~2 of~\cite{my} to extract a triple; however, the original running time is a bit too much. We circumvent this obstacle in the following manner. As $X$ is a $(k,4)$-short jungle, then if we define $S$ to be the subdigraph induced in $T$ by $X$ and, for every pair $v,w$ of vertices in $X$, $k$ paths of length at most $4$ from $v$ to $w$, then $X$ is still a $(k,4)$-short jungle in $S$, but $S$ has size $O(k^3)$. As we store the short jungle together with the corresponding family of paths between the vertices, we can construct $S$ in $O(k^{O(1)})$ time and, using Lemma~2 of~\cite{my}, in $O(k^{O(1)})$ time find a $p(|H|)$-triple inside $S$. This $p(|H|)$-triple is of course also a $p(|H|)$-triple in $T$. We apply Lemma~$6$ of~\cite{my} to find an irrelevant vertex in this triple in $O(p(|H|)^2|V(T)|^2)$ time, delete it, and restart the algorithm.

As there are $|V(T)|$ vertices in the graph, the algorithm makes at most $|V(T)|$ iterations. As every iteration takes $O(h(|H|)|V(T)|^2)$ time for some elementary function $h$, the claim follows.
\end{proof}

\section{Conclusions}\label{sec:conclusions}

The natural question stemming from this work is whether the new set of obstacles can give raise to more powerful irrelevant vertex rules. For example, if we consider the {\sc{Rooted Immersion}} problem, it is tempting to try to replace finding an irrelevant vertex in a triple, as presented in~\cite{my} (see this work for an exact definition of a triple), with a direct irrelevant vertex rule on a short jungle of size polynomial in the size of the digraph to be immersed. If this was possible, the running time for the algorithm for {\sc{Rooted Immersion}} could be trimmed to single-exponential in terms of the size of the digraph to be immersed.

\subparagraph*{Acknowledgements} The author thanks Marek Cygan, Fedor V. Fomin, and Marcin Pilipczuk for helpful discussions about this project.

\bibliographystyle{plain}
\bibliography{pathwidth}

\newpage

\appendix

\section{Dynamic programming routine for topological containment}\label{app:dp}

First, we need to introduce definitions that will enable us to encode all the possible interactions between a model of a digraph $H$ and a separation. Let $(A,B)$ be a separation of $T$, where $T$ is a given semi-complete digraph.

In the definitions we use two special symbols: $\FF,\UU$; the reader can think of them as an arbitrary element of $A\setminus B$ ({\emph{forgotten}}) and $B\setminus A$ ({\emph{unknown}}), respectively. Let $\iota:V(T)\to (A\cap B)\cup \{\FF,\UU\}$ be defined as follows: $\iota(v)=v$ if $v\in A\cap B$, whereas $\iota(v)=\FF$ for $v\in A\setminus B$ and $\iota(v)=\UU$ for $v\in B\setminus A$.

\begin{definition}[\textbf{Traces and signatures}] Let $P$ be  a path.
A sequence of paths $(P_1,P_2,\ldots,P_h)$ is a {\emph{\ptrace}} of $P$ with respect to $(A,B)$, if $P_i$ for $1\leq i\leq h$ are all maximal subpaths of $P$ that are fully contained in $A$, and the indices in the sequence reflect their ordering on path $P$.

Let $(P_1,P_2,\ldots,P_h)$ be the \ptrace{} of $P$ with respect to $(A,B)$. A {\emph{signature}} of $P$ on $(A,B)$ is a sequence of pairs $((b_1,e_1),(b_2,e_2),\ldots,(b_h,e_h))$, where $b_h,e_h\in (A\cap B)\cup\{\FF\}$, such that for every $i\in\{1,2,\ldots,h\}$:
\begin{itemize}
\item $b_i$ is the beginning of path $P_i$ if $b_i\in A\cap B$, and $\FF$ otherwise;
\item $e_i$ is the end of path $P_i$ if $e_i\in A\cap B$, and $\FF$ otherwise.
\end{itemize}
\end{definition}

In other words,  $b_i,e_i$ are images of the beginning and the end of the path $P_i$ in the mapping $\iota$. Observe following properties of the introduced notion:
\begin{itemize}
\item Signature of path $P$ on separation $(A,B)$ depends only on its \ptrace; therefore, we can also consider signatures of \ptraces.
\item It can happen that $b_i=e_i\neq \FF$ only if $P_i$ consists of only one vertex $b_i=e_i$.
\item From the definition of separation it follows that only for $i=h$ it can happen that $e_i=\FF$, as there is no arc from $A\setminus B$ to $B\setminus A$.
\item The empty signature corresponds to $P$ entirely contained in $B\setminus A$.
\end{itemize}

Now we are able to encode relevant information about a given model of $H$.

\begin{definition}
Let $\eta$ be an expansion of a digraph $H$ in $T$. A {\emph{signature}} of $\eta$ on $(A,B)$ is a mapping $\rho$ such that:
\begin{itemize}
\item for every $v\in V(H)$, $\rho(v)=\iota(\eta(v))$;
\item for every $e\in E(H)$, $\rho(e)$ is a signature of $\eta(e)$.
\end{itemize}
\end{definition}

Observe that non-forgotten beginnings and endings in all the signatures of paths from $\rho(E(H))$ can be equal only if they are in the same pair or correspond to a vertex from $\rho(V(H))$. Therefore, we have the following claim.

\begin{lemma}\label{lem:signature-bound}
If $|V(H)|=k, |E(H)|=\ell, |A\cap B|=m$, then the number of possible different signatures on $(A,B)$ is at most
$$(m+2)^k \cdot (m+2)^m\cdot m^\ell\cdot m!\cdot (m+2)^\ell=2^{O((k+\ell+m)\log m)}.$$
Moreover, all of them can be enumerated in $2^{O((k+\ell+m)\log m)}$ time.
\end{lemma}
\begin{proof}
The consecutive terms correspond to:
\begin{enumerate}
\item the choice of mapping $\rho$ on $V(H)$;
\item for every element of $(A\cap B)\setminus \rho(V(H))$, choice whether it will be the end of some subpath in some path signature, and in this case, the value of corresponding beginning (a vertex from $A\cap B$ or $\FF$);
\item for every pair composed in such manner, choice to which $\rho(e)$ it will belong;
\item the ordering of pairs along the path signatures;
\item for every $(v,w)\in E(H)$, choice whether to append a pair of form $(b,\iota(\rho(w)))$ at the end of the signature $\rho((v,w))$, and in this case, the value of $b$ (a vertex from $A\cap B$ or $\FF$).
\end{enumerate}
It is easy to check that using all these information one can reconstruct the whole signature. For every object constructed in the manner above we can check in time polynomial in $k,\ell,m$, whether it corresponds to a possible signature. This yields the enumeration algorithm.
\end{proof}

The set of possible signatures on separation $(A,B)$ will be denoted by $\vf{A}{B}$. We are now ready to present the dynamic programming routine.

\begin{lemma}\label{lem:dp}
There exists an algorithm, which given a digraph $H$ with $k$ vertices and $\ell$ edges, and a semi-complete digraph $T$ together with its path decomposition of width $p$, checks whether $H$ is topologically contained in $T$ in time $O(2^{O((p+k+\ell)\log p)}|V(T)|)$.
\end{lemma}
\begin{proof}
Let $W = (W_1, \dots , W_r)$ be a path decomposition of $T$ of width $p$. Without loss of generality we assume that  $W$ is a nice path decomposition.

By Lemma~\ref{lem:signature-bound}, for every separation $(A,B)=(\bigcup_{j=1}^i W_j, \bigcup_{j=i}^r W_j)$ with separator $W_j$ the number of possible signatures is $2^{O((p+k+\ell)\log p)}$. We will consecutively compute the values of a binary table $D_{(A,B)}:\vf{A}{B}\to \{\bot,\top\}$ with the  following meaning. For $\rho\in \vf{A}{B}$, $D_{(A,B)}[\rho]$ tells, whether there exists a mapping $\overline{\rho}$ with the following properties:
\begin{itemize}
\item for every $v\in V(H)$, $\overline{\rho}(v)=\rho(v)$ if $\rho(v)\in (A\cap B)\cup \{\UU\}$ and $\overline{\rho}(v)\in A\setminus B$ if $\rho(v)=\FF$;
\item for every $e=(v,w)\in E(H)$, $\overline{\rho}(e)$ is a correct path trace with signature $\rho(e)$, beginning in $\overline{\rho}(v)$ if $\overline{\rho}(v)\in A$ and anywhere in $A$ otherwise, ending in $\overline{\rho}(w)$ if $\overline{\rho}(w)\in A$ and anywhere in $A\cap B$ otherwise;
\item path traces $\overline{\rho}(e)$ are internally vertex-disjoint.
\end{itemize}
Such mapping $\overline{\rho}$ will be called a {\emph{partial expansion}} of $H$ on $(A,B)$.

For the first separation $(\emptyset,V(T))$ we have exactly one signature with value $\top$, being the signature which maps all the vertices into $\UU$ and all the arcs into empty signatures. The result of the whole computation should be the value for the signature for the last separation $(V(T),\emptyset)$, which maps all vertices into $\FF$ and arcs into signatures consisting of one pair $(\FF,\FF)$. Therefore, it suffices to show how to fill the values of the table for {\bf{introduce vertex}} step and {\bf{forget vertex}} step.

\paragraph*{Introduce vertex step}

Let us introduce vertex $v\in B\setminus A$ to the separation $(A,B)$, i.e., we consider the new separation $(A\cup\{v\},B)$. Let $\rho\in \vf{A\cup\{v\}}{B}$. We show that $D_{(A\cup\{v\},B)}[\rho]$ can be computed by analyzing the way signature $\rho$ interferes with the vertex $v$.

\begin{itemize}
\item[{\emph{Case $1$:}}] $v\notin \rho(V(H))$, that is, $v$ is not an image 
of  a  vertex of $H$.
\begin{itemize}
\item[{\emph{Case $1.1$:}}] $b_i=v=e_i$ for some pair $(b_i,e_i)\in \rho(e)$ and some $e\in E(H)$. This means that the signature of the partial expansion truncated to separation $(A,B)$ must look exactly like $\rho$, but without this subpath of length zero. Thus $D_{(A\cup\{v\},B)}[\rho]=D_{(A,B)}[\rho']$, where $\rho'$ is constructed from $\rho$ by deleting this pair from the corresponding signature.
\item[{\emph{Case $1.2$:}}] $b_i=v\neq e_i$ for some pair $(b_i,e_i)\in \rho(e)$ and some $e\in E(H)$. This means that the partial expansion truncated to separation $(A,B)$ has to look the same but for the path corresponding to this very pair, which needs to be truncated by vertex $v$. The new beginning has to be either a vertex in $A\cap B$, or a forgotten vertex from $A\setminus B$. As $T$ is semi-complete and $(A,B)$ is a separation, there is an arc from $v$ to every vertex of $A\setminus B$. Therefore, $D_{(A\cup\{v\},B)}[\rho]=\bigvee_{\rho'} D_{(A,B)}[\rho']$, where the disjunction is taken over all signatures $\rho'$ such that in $\rho'$ the pair $(b_i,e_i)$ is substituted with $(b_i',e_i)$, where $b_i'=\FF$ or $b_i'$ is any vertex of $A\cap B$ such that there is an arc $(v,b_i')$.
\item[{\emph{Case $1.3$:}}] $b_i\neq v=e_i$ for some pair $(b_i,e_i)\in \rho(e)$ and some $e\in E(H)$. Similarly as before, partial expansion truncated to separation $(A,B)$ has to look the same but for the path corresponding to this very pair, which needs to be truncated by vertex $v$. As $(A,B)$ is a separation, the previous vertex on the path has to be in the separator $A\cap B$. Therefore, $D_{(A\cup\{v\},B)}[\rho]=\bigvee_{\rho'} D_{(A,B)}[\rho']$, where the disjunction is taken over all signatures $\rho'$ such that in $\rho'$ the pair $(b_i,e_i)$ is substituted with $(b_i,e_i')$, where $e_i'$ is any vertex of $A\cap B$ such that there is an arc $(e_i',v)$.
\item[{\emph{Case $1.4$:}}] $v$ is not contained in any pair in any signature from $\rho(E(H))$. Either $v$ lies on some path in the partial expansion, or it does not. In the first case the corresponding path in partial expansion on $(A\cup\{v\},B)$ has to be split into two subpaths, when truncating the expansion to $(A,B)$. Along this path, the arc that was used to access $v$ had to go from inside $A\cap B$, due to $(A\cup\{v\},B)$ being a separation; however, the arc used to leave $v$ can go to $A\cap B$ and to any forgotten vertex from $A\setminus B$, as $(A,B)$ is a separation and $T$ is a semi-complete digraph. In the second case, the signature of the truncated expansion stays the same. Therefore, $D_{(A\cup\{v\},B)}[\rho]=D_{(A,B)}[\rho]\vee \bigvee_{\rho'} D_{(A,B)}[\rho']$, where the disjunction is taken over all signatures $\rho'$ such that in $\rho'$ exactly one pair $(b_i,e_i)$ is substituted with two pairs $(b_i,e_i')$ and $(b_i',e_i)$, where $e_i'\in A\cap B$ with $(e_i',v)\in E(T)$, whereas $b_i'=\FF$ or $b_i'\in A\cap B$ with $(v,b_i')\in E(T)$.
\end{itemize}
\item[{\emph{Case $2$:}}] $v=\rho(u)$ for some $u\in V(H)$. For every $(u,u')\in E(H)$, $v$ has to be the beginning of the first pair of $\rho((u,u'))$; otherwise, $D_{(A\cup\{v\},B)}(\rho)=\bot$. Similarly, for every $(u',u)\in E(H)$, $v$ has to be the end of the last pair of $\rho((u',u))$; otherwise, $D_{(A\cup\{v\},B)}(\rho)=\bot$. Therefore, $D_{(A\cup\{v\},B)}[\rho]=\bigvee_{\rho'} D_{(A,B)}[\rho']$, where the disjunction is taken over all signatures $\rho'$ such that the first pairs of all $\rho((u,u'))$ are truncated as in the case $1.2$ for all $(u,u')\in E(H)$ (or as in case $1.1$, if the beginning and the end coincide), while the last pairs of all $\rho((u',u))$ are truncated as in the case $1.3$ for all $(u',u)\in E(H)$ (or as in case $1.1$, if the beginning and the end coincide). Moreover, we impose a condition that $\rho'(u)=\UU$.
\end{itemize}

\paragraph*{Forget vertex step}

Let us forget vertex $w\in A\setminus B$ from the separation $(A\cup \{w\},B)$, i.e., we consider the new separation $(A,B)$. Let $\rho\in \vf{A}{B}$; we argue that $D_{(A,B)}[\rho]=\bigvee_{\rho'\in \Gg} D_{(A\cup\{w\},B)}[\rho']$ for some set $\Gg\subseteq \vf{A\cup\{w\}}{B}$, which corresponds to possible extensions of $\rho$ to the previous, bigger separation. We now discuss, which signatures $\rho'$ are needed in $\Gg$ by considering all the signatures $\rho'\in\vf{A\cup\{w\}}{B}$ partitioned with respect to behaviour on vertex $w$.

\begin{itemize}
\item[{\emph{Case $1$:}}] $w\notin \rho'(V(H))$, that is, $w$ is not in the image of $V(H)$.
\begin{itemize}
\item[{\emph{Case $1.1$:}}] $b_i=w=e_i$ for some pair $(b_i,e_i)\in \rho'(e)$ and some $e\in E(H)$. This means that in the corresponding partial expansions $w$ had to be left to $B\setminus A$; however, in $w$ there is no arc from $w$ to $B\setminus A$. Therefore, in $\Gg$ we consider no signatures $\rho'$ of this type.
\item[{\emph{Case $1.2$:}}] $b_i=w\neq e_i$ for some pair $(b_i,e_i)\in \rho'(e)$ and some $e\in E(H)$. If we are to consider $\rho'$ in $\Gg$, then $w$ must prolong some path from the signature $\rho$ in such a manner that $w$ is its beginning. After forgetting $w$ the beginning of this path belongs to the forgotten vertices; therefore, in $\Gg$ we consider only signatures $\rho'$ which differ from $\rho$ on exactly one pair: in $\rho'$ there is $(w,e_i)$ instead of $(\FF,e_i)$ in $\rho$.
\item[{\emph{Case $1.3$:}}] $b_i\neq w=e_i$ for some pair $(b_i,e_i)\in \rho'(e)$ and some $e\in E(H)$. If we are to consider $\rho'$ in $\Gg$, then $w$ must prolong some path from the signature $\rho$ in such a manner, that $w$ is its end. As $w\notin\rho'(V(H))$, after forgetting the vertex $w$ the path should still end in the separator. We obtain a contradiction; therefore, we take no such signature under consideration in the set $\Gg$.
\item[{\emph{Case $1.4$:}}] $w$ is not contained in any pair in any signature from $\rho'(E(H))$. In this case $w$ has to be simply some unused vertex or some vertex inside a path; therefore, we take under consideration in $\Gg$ exactly one signature $\rho'=\rho$.
\end{itemize}
\item[{\emph{Case $2$:}}] $w=\rho'(u)$ for some $u\in V(H)$. We proceed similarly to cases $1.2$ and $1.3$. We consider in $\Gg$ signatures $\rho'$ that differ from $\rho$ in following manner: $\rho'(u)=w$ where $\rho(u)=\FF$ for exactly one $u\in V(H)$; for all edges $(u,u')\in E(H)$ the first pair of $\rho'((u,u'))$ is of form $(w,e_1)$, whereas the first pair of $\rho((u,u'))$ is of form $(\FF,e_1)$; for all edges $(u',u)\in E(H)$ the last pair of $\rho'((u',u))$ is of form $(b_h,w)$, whereas the last pair of $\rho((u,u'))$ is of the form $(b_h,\FF)$.
\end{itemize}

Updating table $D_{(A,B)}$ for each separation requires at most
$ O(|\vf{A}{B}|^2\cdot (p+k+\ell)^{O(1)})= 2^{O((p+k+\ell)\log p)}$ time, and since the number of separations in the pathwidth decomposition is $O(|V(T)|)$, the theorem follows.
\end{proof}

\end{document}